\documentclass[%
pra,
twocolumn,
superscriptaddress,
amsmath,amssymb,groupaddress,longbibliography,
aps,
]{revtex4-2}
\usepackage[colorlinks=true,linkcolor=blue,urlcolor=blue,citecolor=blue,pdfusetitle]{hyperref}
 
\usepackage{physics,xcolor,palatino}
\usepackage{graphicx}% Include figure files
\usepackage{dcolumn}% Align table columns on decimal point
\usepackage{bm}% bold math
\usepackage[normalem]{ulem}
 
\newcommand{\Id}{\mathbb{I}}
\newcommand{\ui}{\mathrm{i}}
\usepackage{mathtools}
\usepackage{tikz}
\usetikzlibrary{decorations.markings,arrows,arrows.meta}
\usepackage{tikz-network}

\newtheorem{theorem}{Theorem}
\newtheorem{corollary}{Corollary}[theorem]
\newtheorem{lemma}{Lemma}
\newtheorem{conjecture}{Conjecture}

\begin{document}
 
%\preprint{APS/123-QED}
 
\title{Symmetric $N\to M$ telecloning and remote quantum state inference}% Force line breaks with \\
 
\author{Adam G. Hawkins}\email{ahawkins06@qub.ac.uk}
\affiliation{School of Computational Sciences, Korea Institute for Advanced Study, Seoul 02455, Korea}
\affiliation{Centre for Quantum Materials and Technologies, School of Mathematics and Physics, Queen’s University Belfast, BT7 1NN Belfast, United Kingdom}
 
\author{Hannah McAleese}
\affiliation{Department of Computer Science, Munster Technological University, Cork, Ireland}
\affiliation{Centre for Quantum Materials and Technologies, School of Mathematics and Physics, Queen’s University Belfast, BT7 1NN Belfast, United Kingdom}
 
\author{Hyukjoon Kwon}
\affiliation{School of Computational Sciences, Korea Institute for Advanced Study, Seoul 02455, Korea}

\date{\today}% It is always \today, today,
 
\begin{abstract}
Teleporting unknown quantum states between distant nodes of a network is a significant feature of quantum communication. Few studies, however, have been conducted on $N \to M$ telecloning, in which $N$ copies of an unknown quantum state are optimally teleported to $M \geq N$ receivers. Previous work requires global POVMs on all copies and auxiliaries; here, we show that symmetric $N \to M$ telecloning can be probabilistically performed using only sequential (or parallel) Bell state measurements (BSMs), allowing the copies to be spatially separated. When successful, the fidelity of each receiver's reduced state saturates the bound imposed by the no-cloning theorem, with the probability of success being independent of $M$. In the case of an `unsuccessful' BSM, we show that the teleportation fidelity is likely to remain high, with the measurement outcome also providing information about the closest Pauli eigenbasis to the unknown state being telecloned. In this sense, each receiver can remotely infer the unknown quantum state with a level of confidence that scales with the number of copies, even if the fidelity of their own reduced state is sub-optimal. We additionally analyze the required resources to ensure the classical-communication fidelity bound is surpassed, both in terms of two-qubit inseparability and varying classes of multipartite entanglement. Surprisingly, we uncover that, at a given iteration of the protocol, pairwise entanglement is not always necessary to increase the teleportation fidelity and is never required to beat the classical bound for $N\geq 2$.
\end{abstract}

\maketitle
\section{Introduction}
Quantum teleportation~\cite{Bennett_1993_teleportation,Pirandola_2015_review,Hu_2023_review} is a phenomenon of quantum information theory whereby Alice may deterministically teleport a single copy of an arbitrary, unknown quantum state to Bob. This protocol relies on shared entanglement, a joint measurement at Alice's lab and local operations and classical communication (LOCC), the latter of which ensures no information can be transferred superluminally. Meanwhile, the $N\to M$ quantum cloning~\cite{Buzek_1996,Buzek_1997,Gisin,Buzek_1998,Zanardi_1998,Werner_1998,Cerf_cont_2000,Scarani_2005,CERF2006455} of an arbitrary, unknown quantum state involves the consumption of $N$ copies to produce $M\geq N$ approximate clones. While Alice is precluded from making perfect clones by the no-cloning theorem~\cite{Dieks_1982,Wootters_1982}, she may use auxiliary qubits and an entangling unitary operation to distribute the quantum information in such a way that the upper bound of the clones' fidelity is saturated; this is labeled \textit{optimal} quantum cloning. Applications of quantum cloning include certain quantum computing tasks~\cite{Galvao_2000,Ricci_2005} as well as state estimation~\cite{Bruss_1998_state_est,Bruss_1999}, the latter of which is a known eavesdropping strategy for BB84-type~\cite{BB84} quantum key distribution protocols~\cite{Bruss_1998_state_est,Bartkiewicz_2013}.

It has been shown that the two aforementioned concepts can be combined to execute a task known as  \textit{quantum telecloning}~\cite{Murao,Dur_1999,Murao_2000,vanLoock_2001,Araneda_2016}. Specifically, an $N\to M$ quantum telecloning protocol aims to optimally teleport $N$ copies of an arbitrary, unknown quantum state to $M\geq N$ receivers in a network. The study of quantum telecloning allows Alice to transport quantum information to a number of receivers without needing to trust an intermediary, and since the information is distributed to multiple nodes, telecloning could prove useful if there is a high probability of node loss when attempting standard teleportation.

While the output fidelity of quantum (tele)cloning is most often \textit{universal} (i.e., the fidelity is equal for all unknown input states), state-dependent quantum cloning has also been considered~\cite{Bruss_1998_dependent,Fan_2014}. Moreover, quantum cloning may be divided into the \textit{symmetric} and \textit{asymmetric} classes. The former describes quantum cloning protocols where the output fidelity is equivalent for each of the $M$ clones, whereas asymmetric cloning may distribute the quantum information unevenly~\cite{Niu_1998,Cerf_2000,Iblisdir_2005,Zhao_2005}. The latter concept has been further extended to $1\to M$ quantum telecloning, such that the teleportation fidelity for each receiver is generally non-equivalent~\cite{Ghiu_2003,Ferraro_2005,Chen_2007,Das_2024}. Meanwhile, the inverse of quantum telecloning, known as \textit{remote information concentration} (that is, where $M$ clones of a quantum state are remotely purified into a single system), is a related topic of study~\cite{Murao_RIC,Yu_RIC,Wang_RIC} alongside quantum broadcasting~\cite{Barnum_1996,Buzek_broadcast}.

Recently, $1\to M$ cloning and telecloning protocols have been experimentally demonstrated for qubits~\cite{Yang_2021,Pelofske_2022,Pelofske_2024_IEEE,Wen_2026} and continuous-variable systems~\cite{Wang_2021,Lou_2024}. Meanwhile a (potentially sub-optimal) port-based version of telecloning has been theorized~\cite{Okada_2025}. Following the original $1\to M$ telecloning protocol~\cite{Murao}, one might suppose that the natural next step would be to develop an $N\to M$ telecloning protocol, and indeed, this was done shortly thereafter in Ref.~\cite{Dur_1999}. However, to the best of our knowledge, this is the only $N\to M$ telecloning protocol that currently exists in the literature --- excluding a continuous-variable $N+N \to M+M$ telecloning protocol producing $M$ anticlones along with the $M$ clones~\cite{Zhang_2008_cont}. Moreover, it has never been implemented experimentally despite relying on the same resource state as the original $1\to M$ protocol. One possible reason for this is that the protocol requires a global positive operator-valued measure (POVM) to be performed on a $2N$-qubit system comprising all $N$ copies and $N$ auxiliary qubits. Not only is this experimentally challenging, but it also requires all the copies (and auxiliaries) to be situated at the same node of a quantum network. Furthermore, finite-parameter POVMs are currently only known for $N=2,3$, with the number of elements scaling rapidly as $N$ increases.

Here, we introduce an alternative, qubit-based $N\to M$ telecloning protocol that uses only sequential (or, equivalently, parallel) Bell state measurements (BSMs) performed on a series of copy-auxiliary systems. An immediate implication of this set-up is that --- unlike the protocol in Ref.~\cite{Dur_1999} --- the copies may be spatially separated and even be situated at distant nodes, allowing for multi-sender collaboration. We prove that a successful protocol will result in a symmetric, universal, and optimal teleportation fidelity for all receivers and for any $N\leq M$. Although our protocol is probabilistic, we show that when the protocol is unsuccessful, Alice and the receivers can infer the unknown state with a level of confidence that scales with $N$ and is uniquely determined by the prior distribution of input states. This means that when the teleportation fidelity may be sub-optimal, the receivers learn about the unknown state. In any case, we show that an `unsuccessful' protocol can simply be considered a state-dependent (yet still symmetric) $N\to M$ telecloning protocol. We also find the surprising result that the probability of success is determined only by the number of copies, being completely independent of the number of receivers. Finally, we examine the necessity and sufficiency (or lack thereof) of differing classes of multipartite entanglement for the telecloning of qubits. We ultimately uncover that, while pairwise entanglement is necessary for the standard $1\to M$ telecloning protocol, it is generally unnecessary for our $N\to M$ protocol when $N\geq 2$.

The remainder of this manuscript is organized as follows. We begin by reviewing the symmetric $1\to M$ and $N\to M$ telecloning of qubits in Sec.~\ref{sec: preliminaries}. Our symmetric $N\to M$ protocol is then outlined in Sec.~\ref{sec: protocol}, followed by proofs of the optimal cloning fidelity and the $M$-independent success probability. In Sec.~\ref{sec: unsuccessful}, we investigate the performance of the unsuccessful protocols, and explore how they may simply be considered state-dependent $N\to M$ telecloning protocols. We then quantify the level of information gain from these protocols, characterize its relationship with the number of copies and show an explicit example with Pauli eigenstates. The dependence of both the $1\to M$ and $N\to M$ protocols on entanglement is then scrutinized in Sec.~\ref{sec: preconditions}, testing the necessity of certain entanglement classes in this multipartite setting for outperforming natural benchmarks. In Sec.~\ref{sec: conclusions} we offer our concluding remarks, before finally providing an outlook on potential future avenues of research in Sec.~\ref{sec: outlook}.

\section{Preliminaries}\label{sec: preliminaries}
\subsection{Symmetric $1\to M$ telecloning}\label{sec: 1 to M telecloning protocol}
The original $1\to M$ quantum telecloning protocol~\cite{Murao} employs optimal universal quantum cloning (UQC) \textit{machines}~\cite{Buzek_1996,Buzek_1997,Buzek_1998}. Specifically, the class of UQC machines of the $N\to M$ type \cite{Gisin,Werner_1998,Zanardi_1998} are machines $U_{NM}$ that take $N$ identical copies of the arbitrary input state $\ket{\varphi}_X=\alpha\ket{0}+\beta\ket{1}$ (with $\abs{\alpha}^2+\abs{\beta}^2=1$) and output $M\geq N$ optimal clones of such a state, each with associated reduced density matrix
\begin{equation}\label{eq: cloned state}
    \rho_{\mathrm{clone}}(\ket{\varphi}_X) = \gamma\ketbra{\varphi}{\varphi}_X + (1-\gamma)\ketbra{\varphi^\perp}{\varphi^\perp}_X\,,
\end{equation}
    with $\braket{\varphi}{\varphi^\perp}=0$ and
\begin{equation}\label{eq: cloning fidelity factor}
    \gamma(N,M) = \frac{M(N+1)+N}{M(N+2)}\,.
    \end{equation} 
This is the qubit case of the qudit generalization derived in Ref.~\cite{Werner_1998}.

The protocol then proceeds as follows. Alice and the $M$ receivers prepare an entangled $2M$-qubit resource state $\ket{\psi_{\mathrm{T}}}$, defined as
\begin{equation}\label{eq: telecloning resource state}
\begin{split}
    \ket{\psi_{\mathrm{T}}} &= \frac{1}{\sqrt{M+1}}\sum_{j=0}^M \ket{D_j^M}_{PA}\ket{D_j^M}_C\\
    % \frac{1}{\sqrt{M+1}}\sum_{j=0}^M \ket{j}_{PA}^{M}\ket{j}_C^M\\
    &=\frac{1}{\sqrt{2}}\left(\ket{0}_P \otimes \ket{\omega_0}_{AC} + \ket{1}_P\otimes\ket{\omega_1}_{AC}\right)\,,
\end{split}
\end{equation}
with
\begin{align}
    \ket{\omega_0}_{AC} &= \sum_{j=0}^{M-1} \sqrt{\frac{2(M-j)}{M(M+1)}} \ket{D_j^{M-1}}_A\ket{D_j^M}_C\\
    % &= \sum_{j=0}^{M-1} \sqrt{\frac{2(M-j)}{M(M+1)}} \ket{j}_A^{M-1}\ket{j}_C^M\\
    \ket{\omega_1}_{AC} &= \sum_{j=0}^{M-1} \sqrt{\frac{2(M-j)}{M(M+1)}} \big|{D_{M-1-j}^{M-1}}\big\rangle_A\ket{D_{M-j}^M}_C\,.
    % \ket{\omega_1}_{AC} &= \sum_{j=0}^{M-1} \sqrt{\frac{2(M-j)}{M(M+1)}} \ket{M-1-j}_A^{M-1}\ket{M-j}_C^M\,.
\end{align}
Here, we define $\ket{D_j^M}$ as the $M$-qubit Dicke state containing exactly $j$ ones, which can be written as
\begin{equation}
    \ket{D_j^M} = \left( \begin{array}{c}
        M \\
        j 
    \end{array} \right)^{-\frac{1}{2}} \sum_i \mathcal{P}_i \left( \ket{0}^{\otimes (M-j)} \otimes \ket{1}^{\otimes j} \right)\,,
\end{equation}
where $\sum_i \mathcal{P}_i ()$ denotes the sum over all distinct permutations of $M-j$ zeros and $j$ ones.

We denote the \textit{port qubit} as $P$, which Alice holds along with the unknown qubit $X$ and the auxiliary system $A=A_1\cdots A_{M-1}$. Meanwhile, the system $C=C_1\cdots C_M$ represents the $M$ receivers. The $(2M+1)$-qubit state vector of the entire system is therefore $\ket{\psi}_{XPAC} = \ket{\varphi}_X \otimes \ket{\psi_{\mathrm{T}}}$, but we can rewrite this state in terms of the Bell basis of the $XP$ compound as
\begin{equation}\label{eq: XPAC pure state}
\begin{split}
    \ket{\psi}_{XPAC} = \frac{1}{2}\bigg[&\ket{\Phi^+}_{XP}\otimes\left(\alpha\ket{\omega_0}_{AC} + \beta\ket{\omega_1}_{AC}\right)\\ 
    &+ \ket{\Phi^-}_{XP}\otimes\left(\alpha\ket{\omega_0}_{AC} - \beta\ket{\omega_1}_{AC}\right) \\
    &+ \ket{\Psi^+}_{XP}\otimes\left(\beta\ket{\omega_0}_{AC} + \alpha\ket{\omega_1}_{AC}\right) \\
    &+ \ket{\Psi^-}_{XP}\otimes\left(\beta\ket{\omega_0}_{AC} - \alpha\ket{\omega_1}_{AC}\right)\bigg]\,,
\end{split}
\end{equation}
with the usual Bell states
\begin{align}
    \ket{\Phi^\pm} &= \frac{1}{\sqrt{2}}\left(\ket{00}\pm\ket{11}\right)\\
    \ket{\Psi^\pm} &= \frac{1}{\sqrt{2}}\left(\ket{01}\pm\ket{10}\right)\,.
\end{align}
Alice then performs a projective measurement on $XP$ in the Bell basis, before then classically communicating the outcome to all receivers. Each receiver then applies an outcome-dependent unitary operation $U^{XP}$ on their respective qubit according to Table~\ref{tab: telecloning}.
\begin{table}[t]
        \centering
        \caption{Summary of the required unitary operation $U^{XP}$ that each receiver must apply to their qubit to obtain an optimal copy of $\ket{\varphi}_X$, depending on Alice's BSM outcome. Here, $\Id$ is the identity matrix and $\sigma_{x,y,z}$ are the standard Pauli matrices.}
        \label{tab: telecloning}
        \begin{tabular}{c||c}
        Alice's BSM Outcome  & Unitary Operation $U^{XP}$ \\
        \hline
        \hline
        $\ket{\Phi^+}_{XP}$     & $\Id$ \\
        $\ket{\Phi^-}_{XP}$     & $ \sigma_z$ \\
        $\ket{\Psi^+}_{XP}$     & $ \sigma_x$ \\
        $\ket{\Psi^-}_{XP}$     & $ \sigma_y$
        \end{tabular}
\end{table}
The reduced state of each receiver is now $\rho_{\mathrm{clone}}(\ket{\varphi}_X)$, with a fidelity $\gamma(1,M)$. The entanglement structure of $\ket{\psi_{\mathrm{T}}}$ is shown in Fig.~\ref{fig: entanglement structure} for $M=3$.

\subsection{Symmetric $N\to M$ telecloning through global POVMs}
In Ref.~\cite{Dur_1999}, a deterministic protocol is introduced whereby Alice can increase the fidelity of each of the receivers' reduced states to $\gamma(N,M)$ if she possesses $N$ copies of the unknown state $\ket{\varphi}_X$, with $2\leq N \leq M$. The principle of this protocol is similar to that of the original $1\to M$ protocol; the resource state is the same, but the nature of Alice's measurement is more complex. Here, a global POVM must be performed on the subsystem comprising all copies $X_1\cdots X_N$, the port qubit $P$, and the first $N-1$ auxiliary qubits $A_1\cdots A_{N-1}$. While this global measurement has been generalized in terms of $N$, such POVMs are defined in terms of continuous parameters, and so require infinitely many bits to classically communicate perfectly. For $N=2,3$, these POVMs have been equivalently defined using a finite number of parameters~\cite{Dur_1999}, thus significantly reducing the classical communication requirements. However, to the best of our knowledge, no generalization in terms of $N$ currently exists for these finite POVMs, and it is not an immediately trivial calculation. Furthermore, the total number of POVM elements scales at least as $N^5$.

Given the global nature of these generalized, many-element measurements made on a system comprising $2N$ qubits, it is clear to see why this protocol has not yet been implemented experimentally~\footnote{We note, however, that the distribution of the multipartite entangled resource state remains one of the most significant hurdles for all telecloning protocols.}. In the next section, we introduce a protocol which resolves this issue by leveraging only projective two-qubit BSMs, at the cost of its deterministic nature. An additional advantage of the following proposed method is that the auxiliary qubits may therefore be spatially separated.
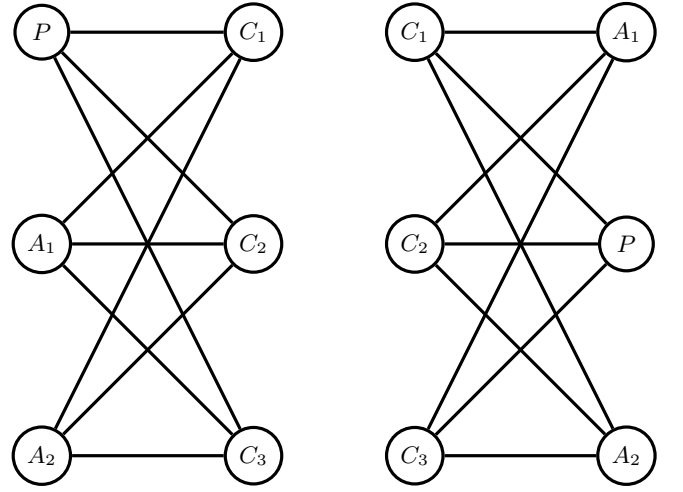
\begin{figure}
    \centering
    \begin{tikzpicture}[node distance={28mm}, very thick, decoration={
    markings,
    mark=at position 0.5 with {\arrow{>}}}, main/.style = {draw, circle}] 
\node[main] (1) {\,$P$\,};
\node[main] (2) [right of=1] {$C_1$};
\node[main] (3) [below of=1] {$A_1$};
\node[main] (4) [below of=2] {$C_2$};
\node[main] (5) [below of=3] {$A_2$};
\node[main] (6) [below of=4] {$C_3$};
\draw[-] (1) -- (2);
\draw[-] (1) -- (4);
\draw[-] (1) -- (6);
\draw[-] (3) -- (2);
\draw[-] (3) -- (4);
\draw[-] (3) -- (6);
\draw[-] (5) -- (2);
\draw[-] (5) -- (4);
\draw[-] (5) -- (6);
\end{tikzpicture}\ \ \ \ \ \ \ \ \ \ \ \ \ \ \ \ \ 
\begin{tikzpicture}[node distance={28mm}, very thick, decoration={
    markings,
    mark=at position 0.5 with {\arrow{>}}}, main/.style = {draw, circle}]
\node[main] (1) {$C_1$};
\node[main] (2) [right of=1] {$A_1$};
\node[main] (3) [below of=1] {$C_2$};
\node[main] (4) [below of=2] {\,$P$\,};
\node[main] (5) [below of=3] {$C_3$};
\node[main] (6) [below of=4] {$A_2$};
\draw[-] (1) -- (2);
\draw[-] (1) -- (4);
\draw[-] (1) -- (6);
\draw[-] (3) -- (2);
\draw[-] (3) -- (4);
\draw[-] (3) -- (6);
\draw[-] (5) -- (2);
\draw[-] (5) -- (4);
\draw[-] (5) -- (6);
\end{tikzpicture}
    \caption{Entanglement structure in the $\ket{\psi_{\mathrm{T}}}$ state for the $M=3$ protocol~\cite{Murao}. The presence of two-qubit entanglement is indicated by a solid black line between nodes. \textbf{Left:} Alice's system $PA$ is shown on the left-hand side, while the receivers $C$ are pictured on the right-hand side. \textbf{Right:} To emphasize the permutation-invariance of the qubits within $\ket{\psi_{\mathrm{T}}}$, we display an alternative arrangement with a different qubit chosen as the port. For any choice of $P$, the qubits on the other side of the diagram become the receivers.}
    \label{fig: entanglement structure}
\end{figure}
\section{Symmetric $N\to M$ telecloning}\label{sec: protocol}
Let us now introduce our symmetric $N\to M$ telecloning protocol, assuming a single sender (though this method is equivalently applicable for up to $N$ senders, each with their own copy of the unknown state). The sender(s) and receivers proceed as follows. (i) Alice and the receivers perform the symmetric $1\to M$ telecloning protocol as outlined in Sec.~\ref{sec: 1 to M telecloning protocol}. (ii) The unitary operations applied to system $C$ (as dictated by Table~\ref{tab: telecloning}) must also be applied to each of the qubits in the auxiliary system $A$. This transforms the reduced state of $AC$ to
\begin{equation}\label{eq: 1 to M transformed state}
\begin{split}
    \ket{\psi_1}_{AC} &= \alpha\ket{\omega_0}_{AC} + \beta\ket{\omega_1}_{AC}\,.
\end{split}
\end{equation}
(iii) Alice then performs BSMs on the subsystems $X_kA_{k-1}$ (with $k\in\{2,\cdots,N\}$), postselecting on outcome $\ket{\Phi^+}_{X_kA_{k-1}}$ for each. Note that, due to the symmetry of the state, it does not matter which of the auxiliary qubits are used for these measurements since they are invariant under permutation, and although we assume sequential BSMs, these measurements commute and can be performed in parallel if so desired. Thus, there is no loss of generality in assuming they are measured in order from $A_1$ to $A_{N-1}$. No local unitary operations are needed for this final step. We show the exemplary $3\to4$ protocol in Fig.~\ref{fig: 3 to 4 protocol}.
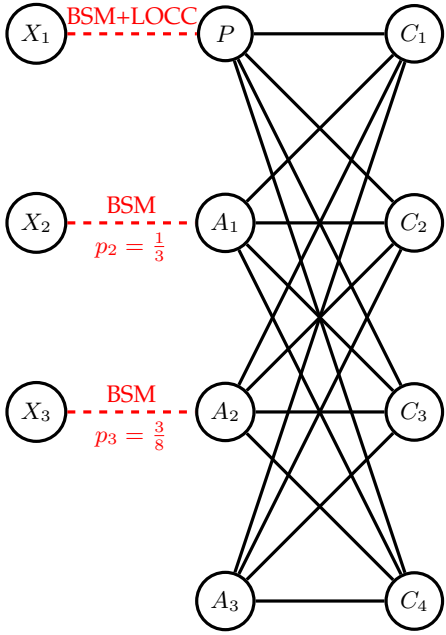
\begin{figure}
\centering
    \begin{tikzpicture}[node distance={25mm}, very thick, decoration={
    markings,
    mark=at position 0.5 with {\arrow{>}}}, main/.style = {draw, circle}] 
\node[main] (1) {\,$P$\,};
\node[main] (2) [right of=1] {$C_1$};
\node[main] (3) [below of=1] {$A_1$};
\node[main] (4) [below of=2] {$C_2$};
\node[main] (5) [below of=3] {$A_2$};
\node[main] (6) [below of=4] {$C_3$};
\node[main] (7) [left of=1] {$X_1$};
\node[main] (8) [below of=7] {$X_2$};
\node[main] (9) [below of=8] {$X_3$};
\node[main] (10) [below of=5] {$A_3$};
\node[main] (11) [below of=6] {$C_4$};
\draw[-] (1) -- (2);
\draw[-] (1) -- (4);
\draw[-] (1) -- (6);
\draw[-] (1) -- (11);
\draw[-] (3) -- (2);
\draw[-] (3) -- (4);
\draw[-] (3) -- (6);
\draw[-] (3) -- (11);
\draw[-] (5) -- (2);
\draw[-] (5) -- (4);
\draw[-] (5) -- (6);
\draw[-] (5) -- (11);
\draw[-] (10) -- (2);
\draw[-] (10) -- (4);
\draw[-] (10) -- (6);
\draw[-] (10) -- (11);
\draw[red, dashed] (7) -- node[midway, above] {\text{BSM+LOCC}} (1);
\draw[red, dashed] (8) -- node[midway, above] {\text{BSM}} node[midway, below] {$p_2=\frac{1}{3}$} (3);
\draw[red, dashed] (9) -- node[midway, above] {\text{BSM}} node[midway, below] {$p_3=\frac{3}{8}$} (5);
\end{tikzpicture}
\caption{Diagrammatic representation of the $3\to4$ telecloning protocol. A red dashed line represents a BSM. After the first BSM, LOCC needs to be performed on each qubit in system $AC$ according to Table~\ref{tab: telecloning} to ensure unit success probability. For the $k$th successful BSM, the resulting fidelity will be $\gamma(k,4)$.}
    \label{fig: 3 to 4 protocol}
\end{figure}

This protocol is clearly probabilistic, since it requires the outcome $\ket{\Phi^+}_{X_k A_{k-1}}$ for the $k$th iteration. Before showing the resulting teleportation fidelity for each receiver, let us first find the success probability. Henceforth, we only denote a BSM as \textit{successful} if an outcome of $\ket{\Phi^+}$ is obtained. Assuming $k-1$ prior successful iterations, the normalized state after the $k$th successful BSM is
\begin{equation}\label{eq: normalized state with general probablity}
\begin{split}
    \ket{\psi_k}_{A^\prime C}=&\left(\prod_{i=1}^{k} \frac{1}{\sqrt{p_i}}\right) \sum_{t=0}^{k}\sum_{j=0}^{M-t} \sqrt{\frac{(M-j)_{t}^\downarrow (j)_{k-t}^\downarrow}{2^{k}(M+1)_{k+1}^\downarrow}}\\
    % &\ \ \ \ \ \ \ \times\begin{pmatrix}
    %     k\\
    %     t
    % \end{pmatrix}
    % \alpha^{k-t}\beta^t\ket{M-t-j}_{A^\prime}^{M-k}\ket{M-j}_C^M\,.
    &\times\begin{pmatrix}
        k\\
        t
    \end{pmatrix}
    \alpha^{k-t}\beta^t\big|{D_{M-t-j}^{M-k}}\big\rangle_{A^\prime}\ket{D_{M-j}^{M}}_C\,,
\end{split}
\end{equation}
where $(x)_{y}^\downarrow=x(x-1)\cdots (x-y+1)$ is a \textit{falling factorial}, $A^\prime$ represents the remaining unmeasured auxiliary qubits ($A_{k}\cdots A_{M-1}$ in this case), and $p_i$ is the probability of the $i$th BSM yielding outcome $\ket{\Phi^+}_{X_i A_{i-1}}$. We include a brief derivation of the above state in Appendix~\ref{app: derivation of state}. We find the surprising result that, not only is the success probability of a given BSM independent of $M$, but also that this probability actually increases with each successful iteration.
\begin{theorem}\label{theorem: probability}
    The probability $p_k$ of observing outcome $\ket{\Phi^+}_{X_kA_{k-1}}$ for the $k$th BSM (assuming all previous BSMs have been successful) is independent of $M$ for all $2\leq k\leq N\leq M$, and is given by
    \begin{equation}
        p_k = \frac{k}{2(k+1)}\,.
    \end{equation}
\end{theorem}

\textbf{Proof.} See Appendix~\ref{app: proof of theorem 1}.

After $k$ successful iterations, the normalized state (after some simplification) is
\begin{equation}\label{eq: k successful state}
\begin{split}
    \ket{\psi_k}_{A^\prime C}=&\sqrt{\frac{k+1}{(M+1)_{k+1}^\downarrow}} \sum_{t=0}^{k}\sum_{j=0}^{M-t} \sqrt{(M-j)_{t}^\downarrow (j)_{k-t}^\downarrow}\\
    % &\ \ \ \ \ \ \ \times\begin{pmatrix}
    %     k\\
    %     t
    % \end{pmatrix}
    % \alpha^{k-t}\beta^t\ket{M-t-j}_{A^\prime}^{M-k}\ket{M-j}_C^M\,.
    &\times\begin{pmatrix}
        k\\
        t
    \end{pmatrix}
    \alpha^{k-t}\beta^t\big|{D_{M-t-j}^{M-k}}\big\rangle_{A^\prime}\ket{D_{M-j}^{M}}_C\,.
\end{split}
\end{equation}
Given the intrinsic symmetry of this state, the reduced states of the receivers are identical; thus, considering only the teleportation fidelity of receiver $C_1$ is sufficient. To find the reduced state of $C_1$ (and thus the fidelity of $C_1$), it is first helpful to show that the protocol is \textit{covariant}. That is to say, we wish to show that replacing an input state $\ket{\varphi}_X=\alpha\ket{0}+\beta\ket{1}$ with $U\ket{\varphi}_X$ (for arbitrary unitary matrix $U$) means that the reduced states of the receivers are also rotated by $U$. Formally, we aim to show that
\begin{equation}
    \ket{\psi_N(U\ket{\varphi})}_{AC} = \left(E_{A} \otimes U_C^{\otimes M}\right)\ket{\psi_N(\ket{\varphi})}_{AC}\,,
\end{equation}
where $E$ is some operator that we are yet to determine.
\begin{lemma}\label{lemma: covariance}
    When successful, the symmetric $N\to M$ telecloning protocol is covariant for all $1\leq N\leq M$.
\end{lemma}

\textbf{Proof.} See Appendix~\ref{app: proof of lemma 2}.

\begin{theorem}\label{theorem: N to M fidelity}
    Upon successful implementation of the protocol for telecloning $N$ copies of an unknown state $\ket{\varphi}_X$ to $M$ receivers, the telecloning fidelity $\mathcal{F}_{C_i}$ of each receiver's reduced state universally saturates the bound imposed by the no-cloning theorem, i.e.,
    \begin{equation}
        \mathcal{F}_{C_i} = \gamma(N,M)\,.
    \end{equation}
\end{theorem}

\textbf{Proof.} With the protocol being covariant (as per Lemma~\ref{lemma: covariance}), the reduced state of $C_1$ is
\begin{equation}
    \rho_{C_1}(U\ket{\varphi}) = U \rho_{C_1}(\ket{\varphi}) U^\dagger\,.
\end{equation}
The fidelity then evaluates to
\begin{equation}
\begin{split}
    \mathcal{F}_{C_1}(U\ket{\varphi}) &= \bra{\varphi}U^\dagger U \rho_{C_1}(\ket{\varphi}) U^\dagger U \ket{\varphi}\\
    &=\bra{\varphi} \rho_{C_1}(\ket{\varphi})  \ket{\varphi} = \mathcal{F}_{C_1}(\ket{\varphi})\,.
\end{split}
\end{equation}
Since any pure state can be represented by $U\ket{\varphi}$ for arbitrary $\varphi$ and $U$, the fidelity must therefore be independent of $\varphi$. We therefore take $\alpha=1,\beta=0$ to significantly simplify the derivation. From Eq.~\eqref{eq: k successful state}, the state after $k$ successful measurements for $\ket{\varphi}=\ket{0}$ is
\begin{equation}
\begin{split}
    % \ket{\psi_k(\ket{0})}_{A^\prime C} &=  \sum_{j=0}^{M}\sqrt{\frac{(k+1)(j)^\downarrow_k}{(M+1)^\downarrow_{k+1}}} \ket{M-j}_{A^\prime}^{M-k}\!\ket{M-j}_C^M\\
    &\ket{\psi_k(\ket{0})}_{A^\prime C} =  \sum_{j=0}^{M}\sqrt{\frac{(k+1)(j)^\downarrow_k}{(M+1)^\downarrow_{k+1}}} \big|{D_{M-j}^{M-k}}\big\rangle_{A^\prime}\ket{D_{M-j}^M}_C\\
    % &= \sum_{j=0}^{M-k}\sqrt{\frac{(k+1)(M-j)^\downarrow_k}{(M+1)^\downarrow_{k+1}}} \ket{j}_{A^\prime}^{M-k}\ket{j}_C^M\,,
    &\ \ \ \ \ \ \ \ \ \ = \sum_{j=0}^{M-k}\sqrt{\frac{(k+1)(M-j)^\downarrow_k}{(M+1)^\downarrow_{k+1}}} \big|{D_j^{M-k}}\big\rangle_{A^\prime}\ket{D_j^{M}}_C
\end{split}
\end{equation}
where the convention $(x)_0^\downarrow=1$ is used. In the second line, we have taken $j\to M-j$ (giving the same summation limits) and subsequently noticed that all factors from $j=M-k+1$ onward are zero, this being a reassuring sanity check since these values are associated with nonphysical states. The density matrix then becomes
\begin{equation}
    \rho_{A^\prime C}^{k}(\ket{0}) = 
    % \sum_{j,j^\prime=0}^{M-k} c_j c_{j^\prime} \ketbra{j}{j^\prime}_{A^\prime}^{M-k} \otimes \ketbra{j}{j^\prime}^M_C\,,
    \sum_{j,j^\prime=0}^{M-k} c_j c_{j^\prime} \big|{D_j^{M-k}}\big\rangle\!\big\langle{D_{j^\prime}^{M-k}}\big|_{A^\prime} \otimes \ketbra{D_j^M}{D_{j^\prime}^M}_C
\end{equation}
with
\begin{equation}
    c_j = \sqrt{\frac{(k+1)(M-j)^\downarrow_k}{(M+1)^\downarrow_{k+1}}}\,.
\end{equation}
Using the Dicke basis $\{\ket{D_q^{M-k}}_{A^\prime}\}_{q=0}^{M-k}$, we can trace over $A^\prime$ to give
\begin{equation}
\begin{split}
    \rho_C^{k}(\ket{0}) 
    % &= \sum_{q,j,j^\prime=0}^{M-k} c_j c_{j^\prime} \braket{q}{j}_{A^\prime}^{M-k}\braket{j^\prime}{q} _{A^\prime}^{M-k} \ketbra{j}{j^\prime}^M_C\\
    &= \sum_{q,j,j^\prime=0}^{M-k} c_j c_{j^\prime} \big\langle{D_q^{M-k}}\big|{D_j^{M-k}}\big\rangle_{A^\prime}\\
    &\ \ \ \ \ \ \ \ \ \ \ \ \ \ \ \times\big\langle{D_{j^\prime}^{M-k}}\big|{D_q^{M-k}}\big\rangle_{A^\prime}\ketbra{D_j^{M}}{D_{j^\prime}^{M}}_C\\
    &= \sum_{j=0}^{M-k} c_j^2 \ketbra{D_j^M}{D_j^M}_C\,.
\end{split}
\end{equation}
Since, for a given value of $j$, the state has $M-j$ zeros and $j$ ones, the reduced state of $C_1$ (and indeed, of any other given receiver) is
\begin{equation}
    \rho_{C_1}^{k}(\ket{0}) = \sum_{j=0}^{M-k} c_j^2\left[\frac{M-j}{M}\ketbra{0}{0}_{C_1} + \frac{j}{M}\ketbra{1}{1}_{C_1}\right]\,,
\end{equation}
giving a fidelity with $\ket{0}$ of
\begin{equation}
    \mathcal{F}_{C_1} = \frac{k+1}{M(M+1)_{k+1}^\downarrow} \sum_{j=0}^{M-k} (M-j)_k^\downarrow (M-j)\,.
\end{equation}
We then notice that we can write $M-j = (M-j-k)+k$ and $(M-j)^\downarrow_{k}\cdot (M-j-k) = (M-j)^\downarrow_{k+1}$ to split the sum as
\begin{equation}
\begin{split}
    \sum_{j=0}^{M-k}&(M-j)^\downarrow_{k}\cdot (M-j)\\
    &= \sum_{j=0}^{M-k}(M-j)^\downarrow_{k+1} + k\sum_{j=0}^{M-k}(M-j)^\downarrow_{k}\\
    &= \sum_{j=0}^{M-k} (k+1)! \begin{pmatrix}
        M-j\\
        k+1
    \end{pmatrix} + k\sum_{j=0}^{M-k}k! \begin{pmatrix}
        M-j\\
        k
    \end{pmatrix} \,.
\end{split}
\end{equation}
After substituting and summing over $l=M-j$, we may employ the \textit{hockey-stick identity}~\cite{Jones_1996}, which is stated as
\begin{equation}
    \sum_{l=u}^v\begin{pmatrix}
        l\\
        u
    \end{pmatrix}
    \equiv 
    \begin{pmatrix}
        v+1\\
        u+1
    \end{pmatrix}
\end{equation}
for $u,v\in \mathbb{N}$ and $v\geq u$. After summing from $l=k$ to $l=M$, we eliminate the $j$-dependence, and so the fidelity simplifies to
\begin{equation}
\begin{split}
    \mathcal{F}_{C_1} &= \frac{1}{M\begin{pmatrix}
        M+1\\
        k+1
    \end{pmatrix}} \left[
    (k+1)\begin{pmatrix}
        M+1\\
        k+2
    \end{pmatrix}
     + k\begin{pmatrix}
        M+1\\
        k+1
    \end{pmatrix}
    \right]\\
    &= \frac{1}{M} \left[
    (k+1)\frac{M-k}{k+2}
     + k
    \right]\\
    &=\frac{M(k+1) + k}{M(k+2)}\\ 
    &= \gamma(k,M)\,.
\end{split}
\end{equation}
Since this is true for all $k\in\{1,\cdots,N\}$, Theorem~\ref{theorem: N to M fidelity} is proved. $\square$

It has recently been shown that optimal $N\to M$ quantum cloning of $d$-dimensional systems (for any $d\geq2$) and optimal $N\to M-N$ state transposition~\cite{Buscemi_2003} are complementary channels~\cite{Brzic_2026}. Thus, the $M-N$ remaining auxiliary qubits in our telecloning protocol are necessarily the output of the optimal $N\to M-N$ transposition channel. In the qubit case, the reduced states of these auxiliaries are \textit{anticlones} of the unknown state $\ket{\varphi}_X$ (up to a local unitary rotation); that is, a noisy version of the state $\ket{\varphi^\perp}_X$. The fidelity of these reduced qubit states with $\ket{\varphi^\perp}_X$ is independent of $M$, and is given by the upper bound of the teleportation fidelity when limited to classical communication (see Sec.~\ref{sec: 1 to M resources}).
 
\section{Unsuccessful outcomes and information gain}\label{sec: unsuccessful}
So far, we have seen that optimal $N\to M$ telecloning can be performed probabilistically. We thus naturally arrive at the question \textit{what happens if a given step of the protocol is unsuccessful?} This section answers this question, both in terms of protocol performance as well as wider implications for remotely inferring $\ket{\varphi}_X$.
\subsection{Performances of unsuccessful outcomes}\label{sec: performances of unsuccessful outcomes}
Following from our definition of a successful BSM, we call the $k$th BSM \textit{unsuccessful} if we observe one of the three outcomes $\ket{\Phi^-}_{X_kA_{k-1}},\ket{\Psi^\pm}_{X_kA_{k-1}}$. We start by looking at the associated probabilities of these outcomes, followed by an analysis of the resulting teleportation fidelities.

Let us start with the outcome $\ket{\Phi^-}_{X_kA_{k-1}}$. Assuming $k-1$ prior successful BSMs, the probability $p_k^{\Phi^-}$ associated with this outcome for the $k$th BSM is
\begin{equation}
    p_{k}^{\Phi^-} =\frac{k \left( \abs{\alpha}^4+ \abs{\beta}^4\right)-2(k-2) \abs{\alpha}^2  \abs{\beta}^2}{2 (k+1)}\ ,\ \forall\, k\geq1\,.
\end{equation}
This probability is clearly dependent on $\ket{\varphi}_X$, and takes a maximum value of $p_k^{\Phi^-}=p_k$ when $\ket{\varphi}_X\in\{\ket{0},\ket{1}\}$. By rewriting the unknown state in the Bloch-sphere parameterization as $\ket{\varphi}_X=\cos\!\frac{\theta}{2}\ket{0}+e^{\ui\phi}\sin\!\frac{\theta}{2}\ket{1}$, we see that the probability is entirely dependent on the polar angle $\theta$ with
\begin{equation}
    p_k^{\Phi^-}(\theta) = \frac{k+1 + (k-1)\cos\!2\theta}{4(k+1)}\,.
\end{equation}
In the limit of $k\to\infty$, this function becomes
\begin{equation}
    \lim_{k\to\infty}p_k^{\Phi^-}(\theta) = \frac{1}{4}(1+\cos\!2\theta)\,,
\end{equation}
and so the probability for polar states converges to $\frac{1}{2}$ (as with $p_k$ for all input states) while the probability of an equatorial state (being the minimum for all $k\geq2$) tends to 0. Since this probability function concentrates around the poles as $k$ increases, a larger $k$ means a greater confidence that both Alice and the receivers can have in $X$ being close to the poles of the Bloch sphere (see Sec.~\ref{sec: quantifying info gain}).

Meanwhile, we see a similar behavior for the probabilities associated with outcomes $\ket{\Psi^+}_{X_kA_{k-1}}$ and $\ket{\Psi^-}_{X_kA_{k-1}}$, but instead for certain equatorial states. These probabilities --- in the Bloch-sphere parameterization --- are given by
\begin{align}\label{eq: prob bloch psi plus}
    p_{k}^{\Psi^+}(\theta,\phi) &= \frac{(k-1)\left( 1+ \cos\!2\phi  -2  \cos^2\!\phi  \cos\!2 \theta  \right) +4}{8 (k+1)}\\
    \label{eq: prob bloch psi minus}
    p_{k}^{\Psi^-}(\theta,\phi) &= \frac{(k-1)\left(1 -\cos\!2\phi  -2  \sin^2\!\phi  \cos\!2 \theta \right) +4}{8 (k+1)}\,.
\end{align}
Both of these probabilities reach a maximum value of $p_k$, corresponding to the eigenstates of the $\sigma_x$ and $\sigma_y$ matrices for outcomes $\ket{\Psi^+}_{X_kA_{k-1}}$ and $\ket{\Psi^-}_{X_kA_{k-1}}$, respectively. The Pauli eigenstates that do not correspond to the maximum probability for an unsuccessful outcome instead have a minimum associated probability on the Bloch sphere surface.

Now, let us inspect the fidelity of the reduced states of the receivers upon a given outcome of the $k$th BSM (assuming all prior BSMs have been successful). For each outcome, we have the normalized states of the full system $A^\prime C$~\footnote{We assert the convention that vanishing binomial coefficients are omitted from the sum. We use this convention throughout this manuscript.} as
\begin{widetext}
\begin{align}
    \begin{split}\label{eq: Phi pm states AC}
    \ket{\psi_k^{\Phi^\pm}}_{A^\prime C}&=\mathcal{A}_{k,M}^{\Phi^\pm}  \sum_{t=0}^{k}\sum_{j=0}^{M-t}\alpha^{k-t}\beta^t\Bigg[\begin{pmatrix}
        k-1\\
        t
    \end{pmatrix} \pm \begin{pmatrix}
        k-1\\
        t-1
    \end{pmatrix}\Bigg]
    \mathcal{B}_{k,M}^{j,t}
    % \ket{M-t-j}_{A^\prime}^{M-k}\ket{M-j}_C^M
    \ket{D_{M-t-j}^{M-k}}_{A^\prime}\ket{D_{M-j}^M}_C
\end{split}\\
\begin{split}\label{eq: Psi pm states AC}
\ket{\psi_k^{\Psi^\pm}}_{A^\prime C} &= \mathcal{A}_{k,M}^{\Psi^\pm} \sum_{t=0}^{k}\sum_{j=0}^{M-t} \left[ \begin{pmatrix}
        k-1\\
        t-1
    \end{pmatrix}\alpha^{k+1-t} \beta^{t-1} \pm \begin{pmatrix}
        k-1\\
        t
    \end{pmatrix}\alpha^{k-1-t}\beta^{t+1} \right]\mathcal{B}_{k,M}^{j,t}
    % \ket{M-t-j}_{A^\prime}^{M-k}\ket{M-j}_C^M\,,
    \big|{D_{M-t-j}^{M-k}}\big\rangle_{A^\prime}\ket{D_{M-j}^M}_C\,,
\end{split}
\end{align}
\end{widetext}
where
\begin{align}
    \mathcal{A}_{k,M}^\Pi = \sqrt{\frac{(k+1)p_k}{(M+1)_{k+1}^\downarrow p_k^{\Pi}}}\ ,\ \mathcal{B}_{k,M}^{j,t} = \sqrt{(M-j)_{t}^\downarrow (j)_{k-t}^\downarrow}
\end{align}
with $p_k^{\Phi^+}=p_k$. Looking at Eq.~\eqref{eq: Phi pm states AC} in particular, it becomes clear as to why the outcome $\ket{\Phi^+}_{X_kA_{k-1}}$ is optimal while the others are generally not. By Pascal's rule $\begin{pmatrix}
        k-1\\
        t
    \end{pmatrix} + \begin{pmatrix}
        k-1\\
        t-1
    \end{pmatrix} = \begin{pmatrix}
        k\\
        t
\end{pmatrix}$, we see that such an outcome preserves the structure of the state as well as the iterability of the protocol, also allowing the teleportation fidelity to remain independent of $\alpha$ and $\beta$. For the unsuccessful outcomes, these features are only exhibited for specific input states. For an analysis of the states in the case of multiple unsuccessful outcomes, see Appendix~\ref{app: multiple outcomes}.
\begin{figure*}
    \centering
    \includegraphics[width=\linewidth]{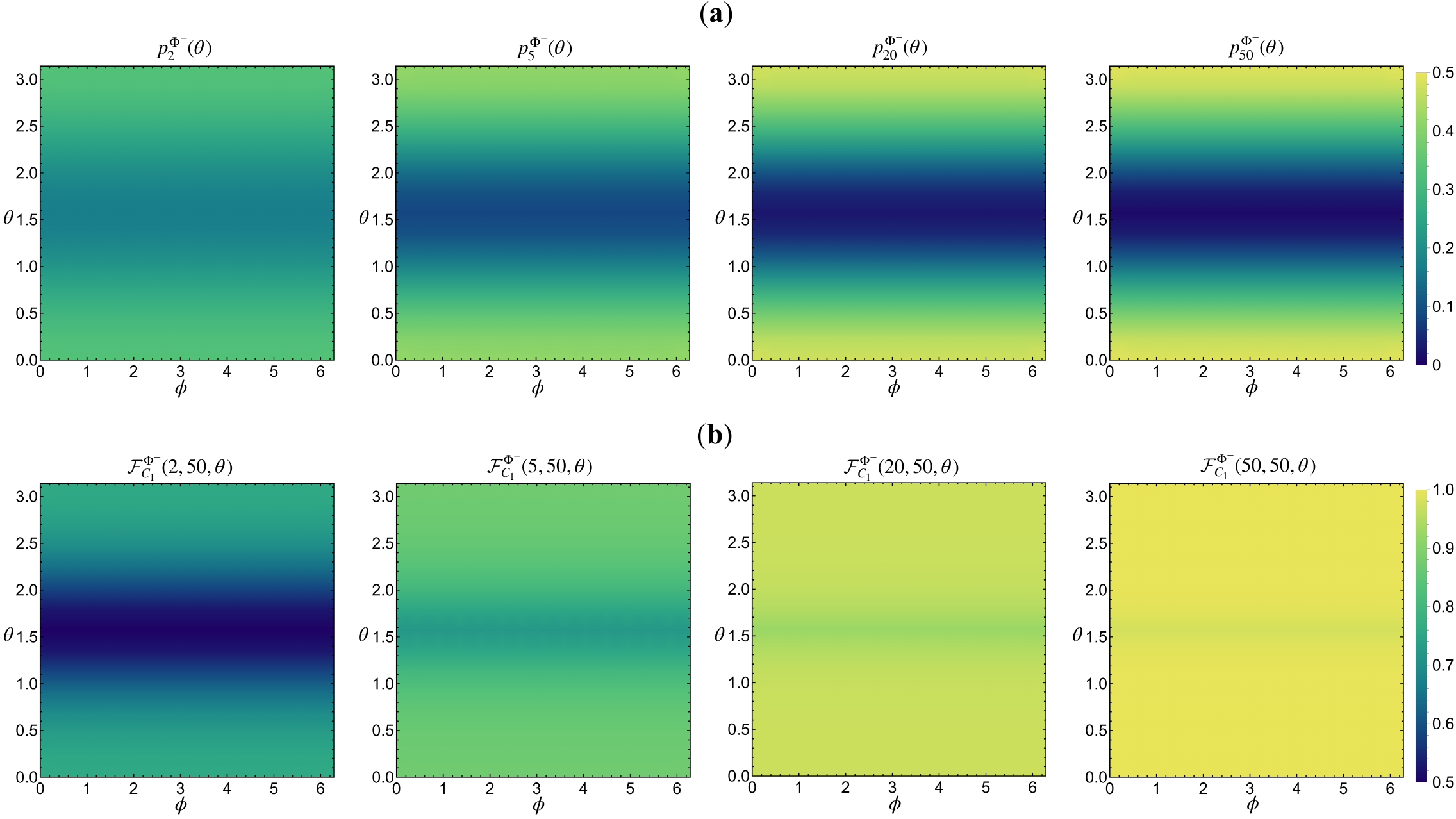}
    \caption{Heatmaps showing the $\ket{\varphi}_X$-dependent probability and fidelity functions --- as a function of the Bloch variables --- when the $k$th BSM outcome is $\ket{\Phi^-}_{X_{k}A_{k-1}}$ (following $k-1$ successful BSMs). Although both functions are independent of $\phi$, we include the azimuth here to help visualize the Bloch sphere and compare to similar plots shown for the $\ket{\Psi^\pm}$ outcomes in Appendix~\ref{app: psi outcomes}. \textbf{(a)} The $M$-independent probability for $k=2,5,20,50$ (from left to right). \textbf{(b)} The fidelity for $M=50$ and $k=2,5,20,50$ (from left to right). We choose a large $M$ such that the progressive state-independence of the fidelity (concurring with the progressive state-dependence of the probability) is well-pronounced. We also notice that the maximum-probability regions always coincide with the maximum-fidelity regions.}
    \label{fig: phi minus heatmaps}
\end{figure*}

We find the $\ket{\varphi}_X$-dependent fidelity functions by tracing over all qubits other than $C_1$ --- a straightforward, yet cumbersome task --- and by then calculating $\mathcal{F}_{C_1}(k,M)=\bra{\varphi}_X \rho_{C_1}^{\tau}(k,M) \ket{\varphi}_X$, where $\rho_{C_1}^{\tau}(k,M)$ is the normalized reduced state of $\rho_{C_1}$ upon BSM outcome $\ket{\tau}_{X_kA_{k-1}}$ (with $\tau\in\{\Phi^\pm,\Psi^\pm\}$). The $\ket{\Phi^-}_{X_kA_{k-1}}$ outcome has an associated fidelity function of
\begin{equation}\label{eq: state dependent phi minus fidelity}
    \mathcal{F}_{C_1}^{\Phi^-}(k,M,\theta) = \frac{(k \mathcal{D}_{k,M}+2) \cos\!2 \theta +(k+2) \mathcal{D}_{k,M}}{M(k+2) [(k-1) \cos\! 2 \theta +k+1]} \,,
\end{equation}
with $\mathcal{D}_{k,M}=k M+k-1$. For any given $\theta$, this function is increasing with $k$, and so for any given input state the fidelity increases with the number of BSMs (and hence the number of copies). For the polar states, Eq.~\eqref{eq: state dependent phi minus fidelity} takes the optimal value of $\mathcal{F}_{C_1}^{\Phi^-}(k,M,0)=\mathcal{F}_{C_1}^{\Phi^-}(k,M,\pi)=\gamma(k,M)$. Indeed, the closer the input state is to the poles, the higher the fidelity, and so we have the useful behavior that the higher the probability of observing $\ket{\Phi^-}_{X_kA_{k-1}}$, the higher the output fidelity of $C_1$ when the outcome $\ket{\Phi^-}_{X_kA_{k-1}}$ is observed.
The minimum value of Eq.~\eqref{eq: state dependent phi minus fidelity} also occurs for the equatorial states, with a value of
\begin{equation}\label{eq: minimum fidelity}
    \mathcal{F}_{C_1}^{\Phi^-}\left(k,M,\frac{\pi}{2}\right)  = \frac{k(M+1)-2}{M(k+2)}\,.
\end{equation}
We may therefore consider an unsuccessful measurement to be an example of a state-dependent quantum telecloning protocol.

Although we have maximum fidelity for the polar states, we see that even the minimum fidelity increases as $k\to M$. In fact, when we have $k\to M$ and $M\to \infty$ concurrently, the fidelity converges to 1 for all possible input states. This implies that, as $k$ increases, the fidelity `flattens' out over the Bloch sphere thus becoming progressively more state-independent, an inverse behavior to the probability function $p_k^{\Phi^-}(\theta)$ which increases in its state dependence. 
We show an extreme example ($M=50$) of how the probability and fidelity functions for outcome $\ket{\Phi^-}_{X_kA_{k-1}}$ evolve with $k$ in Fig.~\ref{fig: phi minus heatmaps}.
The behaviors of the state-dependent fidelity functions for the outcomes $\ket{\Psi^\pm}_{X_kA_{k-1}}$ are analogous to those here, which we detail in Appendix~\ref{app: psi outcomes}.

Having established the state-dependent fidelity functions, we now look to the average resulting fidelity from a given BSM, based on a uniformly random distribution of input states.
\begin{theorem}\label{theorem: average fidelity}
    For a given receiver $C_i$, the average resulting fidelity $\Tilde{\mathcal{F}}_{C_i}(k,M)$ of a given BSM is equal to the optimal postselected fidelity of the previous one, i.e.,
    \begin{equation}
        \Tilde{\mathcal{F}}_{C_i}(k,M) = \gamma(k-1,M)\,.
    \end{equation}
\end{theorem}

\textbf{Proof.} Since the states of the receivers are invariant under permutation change, we consider the fidelity of $C_1$ without loss of generality. To find the average teleportation fidelity $\Tilde{\mathcal{F}}_{C_1}(k,M)$ for uniformly random input states $X$, we weight each of the fidelity functions by their respective probability functions and then average over the surface of the Bloch sphere. This gives us the average probability-weighted fidelity functions
\begin{align}
    \Tilde{\mathcal{F}}_{C_1}^{\Phi^-}(k,M) = \Tilde{\mathcal{F}}_{C_1}^{\Psi^\pm}(k,M)= \frac{k [(k+3) M+k+2]-4}{6 (k+1) (k+2) M}\,,
\end{align}
for the specific outcomes $\ket{\Phi^-}_{X_kA_{k-1}}$ and $\ket{\Psi^\pm}_{X_kA_{k-1}}$, respectively. We then have
\begin{equation}
\begin{split}
    \Tilde{\mathcal{F}}(k&,M) = p_k\gamma(k,M) + \Tilde{\mathcal{F}}_{C_1}^{\Phi^-}(k,M) + 2\Tilde{\mathcal{F}}_{C_1}^{\Psi^\pm}(k,M)\\
    &= \frac{k M+k-1}{(k+1) M} = \gamma(k-1,M)\,.
\end{split}
\end{equation}
As long as the previous $k-1$ steps have been successful, the $k$th BSM perfectly preserves the average fidelity. $\square$

An implication of Theorem~\ref{theorem: average fidelity} is that it helps to answer the question that Alice and the receivers may ask: \textit{if we have had $k<N$ successful BSMs thus far, should we stop and accept the fidelity of $\gamma(k,M)$ as `good enough', or continue and potentially lower our fidelity?} Since the average fidelity of the next BSM is equal to the current fidelity, there is a low risk factor associated with pressing on for higher fidelities, especially since the probability of universal optimal fidelity improves with $k$ (as per Theorem~\ref{theorem: probability}).

\subsection{Quantifying information gain}\label{sec: quantifying info gain}
Given that the probability of observing each unsuccessful outcome is dependent on the unknown state $\ket{\varphi}_X$, observing one of said outcomes will provide Alice with some information about the nature of $\ket{\varphi}_X$ (which she may then classically communicate to each receiver). That is to say, an unsuccessful outcome may suggest to Alice to which eigenbasis the unknown state is probably closest. For example, an outcome of $\ket{\Phi^-}_{X_kA_{k-1}}$ suggests that $\ket{\varphi}_X$ is close to the computational basis. Furthermore in this example, since the probabilities associated with the computational basis states increase with $k$ while the probability associated with the equatorial states simultaneously decreases, one can surmise that the level of information gain also increases with $k$.

To demonstrate this analytically, we choose the Kullback--Leibler (KL) divergence (also called the \textit{relative entropy}) as a suitable quantifier of information gain. For the $k$th outcome $\Pi_k$, this is given by
\begin{equation}
    D_{\mathrm{KL}}^{\Pi} (k)  = \int_{0}^{2\pi}\int_0^\pi p(\theta,\phi\vert \Pi_k)\log_2\!\frac{p(\theta,\phi\vert \Pi_k)}{p(\theta)} \mathrm{d}\theta \mathrm{d}\phi\,.
\end{equation}
Here, $p(\theta)$ is the probability density associated with the unknown state being defined by a given pair of angles $\theta,\phi$. Given that we are dealing with points on the surface of a unit sphere, this probability is independent of $\phi$ since the azimuth does not determine the density of states at a given surface element, and is given by
\begin{equation}
    p(\theta) = \frac{1}{4\pi}\sin\!\theta\,.
\end{equation}
Meanwhile, $p(\theta,\phi\vert \Pi_k)$ is the conditional probability of the input state being defined by $\theta,\phi$, given that we have observed the outcome $\Pi_k$. Currently, we only know the converse probabilities $p(\Pi_k\vert \theta,\phi)\in\{p_k,p_k^{\Phi^-},p_k^{\Psi^+},p_k^{\Psi^-}\}$, but we may employ Bayes' Theorem as
\begin{equation}
    p(\theta,\phi\vert\Pi_k) = \frac{p(\Pi_k\vert \theta,\phi) p(\theta)}{p(\Pi_k)}\,,
\end{equation}
where
\begin{equation}
    p(\Pi_k) = \int_{0}^{2\pi}\int_{0}^{\pi}p(\Pi_k\vert \theta,\phi)p(\theta)\mathrm{d}\theta \mathrm{d}\phi
\end{equation}
giving
\begin{equation}
    p(\Phi^-)=p(\Psi^+)= p(\Psi^-) = 
    \frac{k+2}{6(k+1)}\,.
\end{equation}

We start by showing that obtaining an outcome $\ket{\Phi^-}_{X_kA_{k-1}}$ alongside $k-1$ successful BSMs leads to an information gain of
\begin{equation}\label{eq: KL divergence uniform}
\begin{split}
    D_{\mathrm{KL}}^{\Phi^-}(k) = \frac{2(k+2)^{-1}}{3\ln\!2}\left[\frac{6\tan^{-1}\!\sqrt{k-1}}{\sqrt{k-1}} - (k+5)\right]+\mathcal{E}_k\,,
\end{split}
\end{equation}
bits, where $\mathcal{E}_k=\log_2\!\frac{3k}{k+2}$. It can also be shown that $D_{\mathrm{KL}}^{\Phi^-}(k)= D_{\mathrm{KL}}^{\Psi^+}(k)= D_{\mathrm{KL}}^{\Psi^-}(k)$. The derivative of this function is positive for all $k\geq 2$, and so the more successful BSMs made in conjunction with an unsuccessful one leads to a larger information gain. That is to say, as $k$ increases, Alice and the receivers can be more confident about the eigenbasis to which $\ket{\varphi}_X$ is probably closest. As expected, we find that as $k\to1$ we have $D_{\mathrm{KL}}^{\Pi}\to 0$ since each outcome has a $\theta$-independent probability of $\frac{1}{4}$ for $k=1$. Meanwhile, $k\to \infty$ gives $D_{\mathrm{KL}}^{\Pi}\to (\ln\!27 -2)/\ln\!8\approx 0.623166$ bits of information.

\subsection{Case study: Pauli eigenstates}\label{sec: Pauli case study}
\begin{figure}
    \centering
    \includegraphics[width=\linewidth]{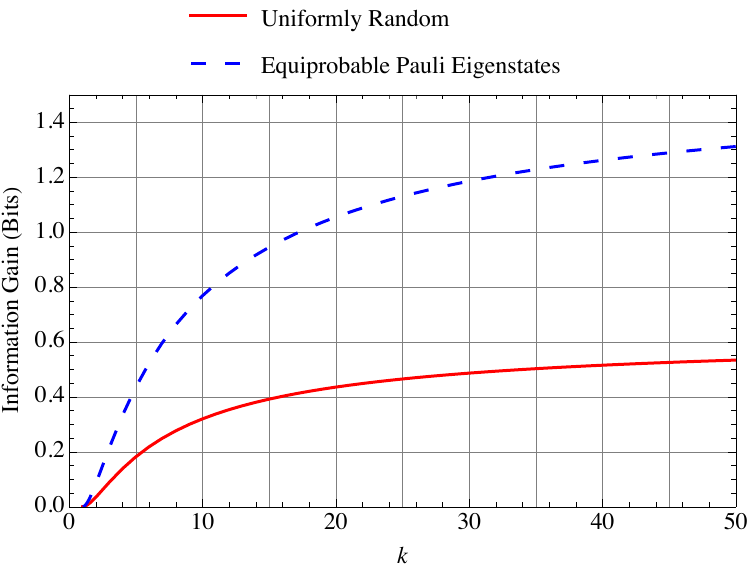}
    \caption{The KL divergence (in bits) plotted against $k$. This is shown for uniformly random input states [$D_{\mathrm{KL}}^{\Pi}(k)$] and equal probability of each of the six Pauli eigenstates [$D_{\mathrm{KL,d}}^{\Pi}(k)$].}
    \label{fig: info gain}
\end{figure}
We have already seen that the maximum (and minimum) values of each of the probability functions $p_k^{\Phi^-},p_k^{\Psi^+},p_k^{\Psi^-}$ are all situated at the points on the Bloch sphere that correspond to the Pauli eigenstates $\ket{0},\ket{1},\ket{+},\ket{-},\ket{+\ui},\ket{-\ui}$, where
\begin{align}
    \ket{\pm} = \frac{1}{\sqrt{2}}\left( \ket{0} \pm \ket{1} \right)\ \ ,\ \ 
    \ket{\pm\ui} = \frac{1}{\sqrt{2}}\left( \ket{0} \pm \ui\ket{1} \right)\,.
\end{align}
Due to this preferential behavior of the unsuccessful outcomes towards these states, one may instead consider their performances with a discrete probability distribution describing the equiprobable input of these six states.

Let us first find the function for information gain. The total probability of each unsuccessful outcome in this scenario, after averaging the probability functions over the six input states, is coincidentally again
\begin{equation}
    p(\Phi^-)=p(\Psi^+)=p(\Psi^-) = \frac{k+2}{6(k+1)}\,.
\end{equation}
The discrete KL divergence $D_{\mathrm{KL,d}}^\Pi$ for outcome $\Pi$ is then
\begin{equation}
    D_{\mathrm{KL,d}}^\Pi(k) = \sum_s p(s \vert \Pi) \log_2\! \frac{p(s \vert \Pi)}{p(s)}\,,
\end{equation}
for input states $s\in\{\ket{0},\ket{1},\ket{+},\ket{-},\ket{+\ui},\ket{-\ui}\}$. The KL divergence is equivalent for each outcome, with
\begin{equation}
    D_{\mathrm{KL,d}}^\Pi(k) = \frac{1}{k+2}\left(2\log_2\!\frac{3}{k+2} + k \log_2\!\frac{3k}{k+2}\right)
\end{equation}
bits.
This function is larger than $D_{\mathrm{KL}}^{\Phi^-}(k)$ for all $k\geq2$, thus confirming more potential for Alice and the receivers to infer the unknown state $\ket{\varphi}_X$. We display both functions in Fig.~\ref{fig: info gain}. As $k\to\infty$, we have $D_{\mathrm{KL,d}}^{\Pi}(k)\to \log_2\!3\approx 1.58496$ bits, significantly more than uniformly random inputs.

If we now average the state-dependent fidelity functions --- weighted by the state-dependent probability functions --- equally over the six input states, we arrive at an average teleportation fidelity of
\begin{equation}
    \Tilde{\mathcal{F}}_{\mathrm{Pauli}}(k,M) = \gamma(k-1,M)\,.
\end{equation}
Thus, although the average fidelity is the same as for uniform randomness (as in Theorem~\ref{theorem: average fidelity}), the discrete probability distribution means that Alice and the receivers stand to gain far more information about the input state.
 
\section{Preconditions for quantum advantage}\label{sec: preconditions}
To supplement the protocol introduced in this manuscript, we explore the interplay between entanglement and the performance of both the symmetric $1\to M$ and $N\to M$ protocols. It has long been established that distillable entanglement~\cite{Horodecki_1998} is necessary for outperforming teleportation protocols of $d$-dimensional systems involving only classical communication~\cite{Horodecki_1999}, and so naturally we now ask similar questions about the quantum telecloning protocols. Since we are dealing with arbitrarily large multipartite states here rather than the bipartite states of standard teleportation, we must consider differing classes of multipartite entanglement and the role played by each in executing the task of telecloning. We first review the delicate differences between categorizations of multipartite entanglement, before probing their necessity (or lack thereof) for symmetric telecloning through appropriate noise models.
  \subsection{Classes of multipartite entanglement: A brief review}
The topic of entanglement in multipartite systems contains many nuanced categories of such entanglement, a few of which we outline here~\cite{Hawkins_thesis}. To fully witness these subtleties, let us start with an $n$-partite pure state $\ket{\psi}_{1,\cdots,n}$. We call this state \textit{fully separable} if it can be written as a product state of all subsystems, i.e.,
\begin{equation}\label{eq: fully separable pure state}
    \ket{\psi}_{1,\cdots, n} = \bigotimes_{i=1}^n \ket{\chi}_i \iff\ \text{fully separable}\,.
\end{equation}
If the pure state cannot be written in this way, then it must contain at least some entanglement. If the state is not fully separable, but rather can only be written as a splitting of $m$ parts $\{P_j\}_{j=1}^m$ in tandem, then it is called \textit{$m$-separable} \cite{Guhne_2009,Sabin_2008}, with a Hilbert space satisfying $\mathcal{H}_1\otimes\cdots\otimes\mathcal{H}_n=\mathcal{H}_{P_1}\otimes\cdots\otimes\mathcal{H}_{P_m}$. That is to say,
\begin{equation}\label{eq: m-separable states}
    \ket{\psi}_{1,\cdots,n} = \bigotimes_{j=1}^m \ket{\chi}_{P_j}\iff\ \text{$m$-separable}\,.
\end{equation}
In the case of $m=2$, the state is said to be \textit{biseparable}. Next, a pure state possesses \textit{genuine multipartite entanglement} (GME) if, and only if, it is neither fully separable nor $m$-separable for any $m\in\{l\}_{l=2}^{n-1}$. We do, however, note that if a state can be written as $m$-separable, then it can also be written as biseparable, and so the definition of a genuinely multipartite entangled state (sometimes called a \textit{fully multipartite entangled state}) is often said to simply be those which are not biseparable \cite{Guhne_2005}. Finally, we have the concept of \textit{global entanglement}. There is no unanimous agreement on the nomenclature in the literature, with some authors using the term `global entanglement' to refer to what we define as GME. In this manuscript, we choose to define a globally entangled state as one that is inseparable with respect to all possible bipartitions of the subsystems, i.e., every bipartition is entangled. Therefore, not only do all globally entangled pure states possess GME, but the conditions for global entanglement and GME are equivalent for pure states.

Let us now see how the above definitions hold up when moving to an $n$-partite mixed state $\rho_{1,\cdots,n}$. The definition of a fully separable multipartite state easily generalizes from Eq.~\eqref{eq: fully separable pure state} as a state which can be written as a convex sum of fully separable $n$-partite pure states $\{\ket{\chi_k^{\mathrm{FS}}}\}$, meaning we can say that
\begin{equation}\label{eq: fully separable mixed state}
    \rho_{1,\cdots,n} = \sum_k r_k \ketbra{\chi_k^{\mathrm{FS}}}{\chi_k^{\mathrm{FS}}} \iff\ \text{fully separable}\,,
\end{equation}
where each state in $\{\ket{\chi_k^{\mathrm{FS}}}\}$ satisfies the condition in Eq.~\eqref{eq: fully separable pure state}. Meanwhile, an $n$-partite, $m$-separable mixed state is one that can be written as a convex combination of pure $m$-separable states (i.e., those which satisfy Eq.~\eqref{eq: m-separable states}) \cite{Guhne_2009}, although we note that each pure state in the mixture may be $m$-separable with respect to different partitions~\cite{Sabin_2008}. We define a mixed state with GME in a similar way to the pure-state case: a state which is neither fully separable nor biseparable, and so cannot be written as a convex combination of biseparable pure states. Since there exist biseparable states which --- although written as a convex sum of biseparable states --- are still inseparable in every possible bipartition, the definitions of GME and global entanglement diverge in the case of mixed states.

To quantify bipartite entanglement, we choose the positive partial transpose (PPT)~\cite{Peres,Horodecki1} \textit{negativity} due to its sufficiency as a measure (and necessity in the case of qubit-qubit and qubit-qutrit systems \cite{Horodecki_arxiv,Horodecki_1998}). The negativity $\mathcal{N}_{Q:R}$ in the $Q:R$ bipartition of the arbitrary density matrix $\rho_{QR}$ (comprising systems $Q$ and $R$) is defined as the absolute value of the sum of the negative eigenvalues of the partial transpose of $\rho_{QR}$~\cite{Vidal_2002}. This can, for instance, be computed via
\begin{equation}\label{eq: negativity definition}
    \mathcal{N}_{Q:R}(\rho_{QR}) = \sum_i \frac{\abs{\lambda_{QR}^i}-\lambda_{QR}^i}{2}\,,
\end{equation}
where $\{\lambda_{QR}^i\}$ are the eigenvalues of the partially transposed matrix $\rho_{QR}^{T_Q}$ (or, alternatively, $\rho_{QR}^{T_R}$).

In an effort to gauge the global entanglement, we use the \textit{$n$-concurrence}. The 2-concurrence was originally introduced as an entanglement measure for a two-qubit state \cite{Wootters_1998}, but was later generalized to an $n$-qubit system where $n$ is even \cite{Wong_2001,Brennen_2004}. Given that we will be considering the resource state $\ket{\psi_{\mathrm{T}}}$ containing $2M$ qubits, this is adequate for our needs. The $n$-concurrence of a state $\rho$ is then computed via
\begin{equation}\label{eq: N-concurrence definition}
    \mathcal{C}_n(\rho) = \max\left[0, \sqrt{\eta_1} - \sum_{j=2}^{2^n}\sqrt{\eta_j}\right]\,,
\end{equation}
where the matrix
\begin{equation}
    \Bar{\rho} = \rho \bigg(\bigotimes_{j=1}^n \sigma_y^j\bigg) \rho^* \bigg(\bigotimes_{j=1}^n \sigma_y^j\bigg)
\end{equation}
has the set of eigenvalues $\{\eta_j\}$ of $\Bar{\rho}$ with $\eta_1\geq\eta_j\forall j$~\cite{Campbell_2010}. Non-zero $\mathcal{C}_n$ is a sufficient indicator of entanglement, with this quantifier also guaranteeing a null value for states which have one qubit separable from all of the rest~\cite{Brennen_2004}. An example of a class of states whose intrinsic entanglement is not detected by this measure is the $n$-partite $W$ states defined in Ref.~\cite{Dur2000}.

\subsection{Resources for symmetric $1\to M$ telecloning}\label{sec: 1 to M resources}
Before exploring the entanglement in the resource state, we first find a suitable threshold against which we benchmark the performance of the protocol. To explore any advantage gained by exploiting the phenomena of quantum theory, it is then natural to find the limit of fidelity when Alice and the receivers are limited to classical communication only. This upper bound, which we denote as $\mathcal{F}_{\mathrm{cl}}$, can be realized by utilizing an \textit{optimal quantum state estimation}~\cite{Massar_1995,Bruss_1999} technique. An exemplary method for doing this is where Alice uses a combination of optimal UQC machines and universal measurements to estimate the state $\ket{\varphi}_X$, before then classically communicating her results to each receiver such that they may recreate the state in their own lab. If Alice possesses $N$ copies of an unknown qubit, the classical-communication fidelity is upper-bounded by
\begin{equation}\label{eq: classical fidelity bound}
    \mathcal{F}_{\mathrm{cl}}(N) = \frac{N+1}{N+2}\,.
\end{equation}

While previous work has gauged the telecloning fidelity using the global entanglement~\cite{Gordon_2007}, we show that pairwise entanglement is the more pertinent entanglement class to the symmetric $1\to M$ protocol.
\begin{theorem}
    Neither global multipartite entanglement nor GME is necessary for quantum advantage (or even optimal performance) of the $1\to M$ quantum telecloning protocol for any number of receivers. That is to say, a telecloning fidelity of $\gamma(1,M)$ does not rely on the precondition of these classes of multipartite entanglement, despite the resource state naturally possessing these properties.
\end{theorem}

\textbf{Proof.} It can be shown that the reduced state of any port-receiver system (we again choose $C_1$ without loss of generality) is
\begin{equation}
    \rho_{PC_1}^M = \frac{1}{6M}\left[(2M+4)\ketbra{\Phi^+}{\Phi^+} + (M-1)\Id\right]\,.
\end{equation}
Performing the standard teleportation protocol using the above state yields a universal teleportation fidelity of
\begin{equation}
    \mathcal{F}_{C_1}^1(M)=\bra{\varphi}_X \rho^{M\prime}_{C_1} \ket{\varphi}_X = \frac{2M+1}{3M} =\gamma(1,M)\,,
\end{equation}
where
\begin{equation}
    \rho_{C_1}^{M\prime}=4\Tr_{XP}\left\{\Pi_{XP}^{\Phi^+}\left[\ketbra{\varphi}{\varphi}_X\otimes \rho_{PC_1}^M\right]\Pi_{XP}^{\Phi^+}\right\}
\end{equation}
with $\Pi_{XP}^{\Phi^+}=\ketbra{\Phi^+}{\Phi^+}_{XP}$. Thus, despite there being pairwise entanglement between each auxiliary qubit in $A$ and each receiver in $C$, the entire auxiliary system can be fully traced out with no effect on the optimality of the protocol. This means that the auxiliary system may be subject to large amounts of noise, and so no resources are required to maintain these two-qubit entanglement links. If we consider even just one of the auxiliary qubits, say $A_1$, to be subject to a very strong local noise channel such that the entanglement in $A_1:PA_2\cdots A_{M-1} C$ dies, this means that optimal teleportation fidelity may still be attained for each receiver despite the resource state no longer possessing global entanglement or GME. $\square$

We have shown in the above proof that the $1\to M$ protocol can be considered a collection of two-qubit teleportation protocols between $P$ and each receiver. Combining this with the known result that distillable entanglement is necessary in such situations~\cite{Horodecki_1999} (and with the fact that all qubit-qubit entangled states are distillable~\cite{Horodecki_distill_1997}), the following corollary is implied.
\begin{corollary}\label{corollary: bipartite}
    Non-zero bipartite entanglement in $P:C_i$ is necessary for quantum advantage when telecloning a single copy of an unknown qubit to the receiver $C_i$ for all $M\geq2$.
\end{corollary}
We show an updated entanglement diagram indicating only the `necessary' pairwise entanglement links in Fig.~\ref{fig: necessary entanglement structure}.
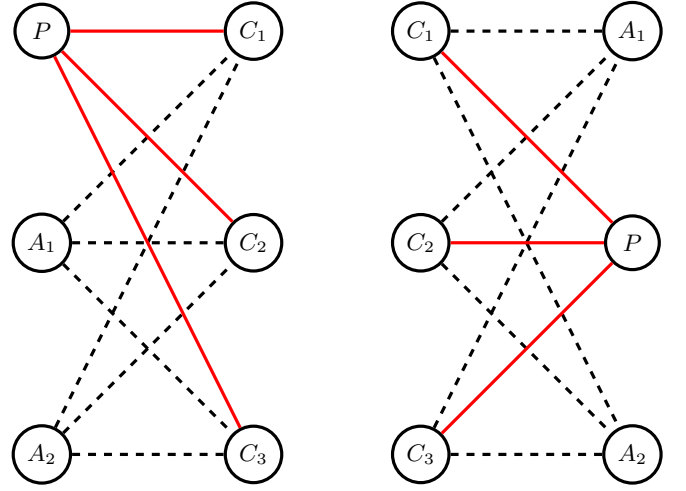
\begin{figure}[t]
    \centering
    \begin{tikzpicture}[node distance={28mm}, very thick, decoration={
    markings,
    mark=at position 0.5 with {\arrow{>}}}, main/.style = {draw, circle}] 
\node[main] (1) {\,$P$\,};
\node[main] (2) [right of=1] {$C_1$};
\node[main] (3) [below of=1] {$A_1$};
\node[main] (4) [below of=2] {$C_2$};
\node[main] (5) [below of=3] {$A_2$};
\node[main] (6) [below of=4] {$C_3$};
\draw[-, dashed] (3) -- (2);
\draw[-, dashed] (3) -- (4);
\draw[-, dashed] (3) -- (6);
\draw[-, dashed] (5) -- (2);
\draw[-, dashed] (5) -- (4);
\draw[-, dashed] (5) -- (6);
\draw[red] (1) -- (2);
\draw[red] (1) -- (4);
\draw[red] (1) -- (6);
\end{tikzpicture}
\ \ \ \ \ \ \ \ \ \ \ \ \ \ \ \ \ 
\begin{tikzpicture}[node distance={28mm}, very thick, decoration={
    markings,
    mark=at position 0.5 with {\arrow{>}}}, main/.style = {draw, circle}] 
\node[main] (1) {$C_1$};
\node[main] (2) [right of=1] {$A_1$};
\node[main] (3) [below of=1] {$C_2$};
\node[main] (4) [below of=2] {\,$P$\,};
\node[main] (5) [below of=3] {$C_3$};
\node[main] (6) [below of=4] {$A_2$};
\draw[-, dashed] (3) -- (2);
\draw[-, dashed] (1) -- (2);
\draw[-, dashed] (3) -- (6);
\draw[-, dashed] (5) -- (2);
\draw[-, dashed] (1) -- (6);
\draw[-, dashed] (5) -- (6);
\draw[red] (4) -- (1);
\draw[red] (4) -- (3);
\draw[red] (4) -- (5);
\end{tikzpicture}
    \caption{A re-working of Fig.~\ref{fig: entanglement structure} to show the necessary entanglement structure in the $\rho_{PAC}$ state for the $1\to3$ protocol. The presence of necessary pairwise entanglement is shown by solid red lines between nodes, while unnecessary pairwise entanglement is indicated by a dashed black line. The symmetric nature of the resource is again exemplified by the freedom of choice of $P$.}
    \label{fig: necessary entanglement structure}
\end{figure}

Since we have now shown that global entanglement is not necessary, one might now inquire about its sufficiency. That is to ask, \textit{if we are assured that there is non-zero global entanglement in the resource state, are we guaranteed quantum advantage for all receivers?} One way we can search for a counter-example is by, e.g., subjecting only the qubits $P$ and $C_1$ to a local noise channel while leaving all other qubits unspoiled. This is such that the negativity in $P:C_1$ may vanish (thus losing quantum advantage for the first receiver as per Corollary~\ref{corollary: bipartite}), yet the global entanglement may still be non-zero due to the rest of the remaining entanglement structure. To this end, we choose local dephasing channels acting on $P$ and $C_1$ defined by the Kraus operators
\begin{equation}
        \mathcal{K}_{0}(\mu){=}\sqrt{1-\mu}\Id\ ,\ \mathcal{K}_{1}(\mu){=}\sqrt{\mu}\ketbra{0}{0}\ ,\ \mathcal{K}_{2}(\mu){=}\sqrt{\mu}\ketbra{1}{1}\,.
\end{equation}
For simplicity, we assume the noise strength $\mu$ to be the same for both qubits, such that the state transforms as
\begin{equation}
    \rho_{PC_1}^{M\prime}(\mu) = \sum_{q,r} [\mathcal{K}_q(\mu)\otimes \mathcal{K}_r(\mu)] \rho_{PC_1}^M [\mathcal{K}_q(\mu)\otimes \mathcal{K}_r(\mu)]^\dagger\,.
\end{equation}
The $M$-dependent level of noise strength $\Tilde{\mu}_{PC_1}(M)$ which induces entanglement sudden death~\cite{Yu,YuReview} is
\begin{equation}
    \Tilde{\mu}_{PC_1}(M) = 1 - \sqrt{\frac{M-1}{M+2}}\,,
\end{equation}
with $\lim_{M\to\infty}\Tilde{\mu}_{PC_1}(M)=0$.

Meanwhile, the $2M$-concurrence of the telecloning resource state with the local dephasing channels acting on $P$ and $C_1$ is
\begin{equation}
    \mathcal{C}_{2M}(M,\mu) = \max\left\{0,\sqrt{(1-\mu)^2+\kappa_M^2\epsilon_\mu^2} - \epsilon_\mu\right\}\,,
\end{equation}
where
\begin{equation}
    \kappa_M=\frac{M+2}{3M}\ \ ,\ \ \epsilon_\mu=\frac{\mu(2-\mu)}{2}\,.
\end{equation}
This function reaches a value of zero at
\begin{equation}
    \Tilde{\mu}_{PAC}(M) = 1 - \frac{\sqrt{17M^2-4M-4}-3M}{2\sqrt{(M-1)(2M+1)}}\,,
\end{equation}
with 
\begin{equation}
    \lim_{M\to\infty}\Tilde{\mu}_{PAC}(M)=\frac{1}{4} \left(4+3 \sqrt{2}-\sqrt{34}\right)\approx0.602922\,.
\end{equation}
We find that $\Tilde{\mu}_{PAC}(M)>\Tilde{\mu}_{PC_1}(M)\ \forall\, M\geq2$, implying that there is always an interval $\mu\in(\Tilde{\mu}_{PC_1},\Tilde{\mu}_{PAC})$ where the $2M$-concurrence is non-zero despite the $P:C_1$ entanglement being zero. Although this is promising evidence, it unfortunately cannot confirm for certain that global entanglement is insufficient due to the limitations of the $n$-concurrence. As mentioned, this is guaranteed to be zero if a single qubit (or indeed, an odd number of qubits) is separable from the rest of the system, but it cannot check bipartitions containing an even number of qubits. Therefore, it is possible that a state has non-zero $n$-concurrence despite there being a separable bipartition. While it may be shown that the negativity in some bipartitions --- such as the potentially relevant $PC_1:AC_2\cdots C_M$ splitting --- is positive for non-zero $\mathcal{C}_{2M}$, checking every possible bipartition would be a grueling task. Therefore, with no simple alternative quantifier at our disposal, we leave this line of inquiry as conjecture.
\begin{conjecture}
    Global entanglement in $\ket{\psi_{\mathrm{T}}}$ is not sufficient for quantum advantage in telecloning a single copy of an unknown quantum state to any given receiver, for any $M\geq2$.
\end{conjecture}

\subsection{Resources for symmetric $N\to M$ telecloning}
Having provided strong evidence that two-qubit pairwise entanglement is the crucial resource for $1\to M$ telecloning, we further inspect its role in our $N\to M$ protocol. In this scenario, there are two criteria for fidelity that make sense to evaluate. The first is to simply increase the fidelity from one iteration to the next; i.e., what do we need to beat a fidelity of $\gamma(k,M)$ for the $(k+1)$th BSM? Meanwhile, the other fidelity threshold we aim to beat is $\mathcal{F}_{\mathrm{cl}}(k+1)$, as per Eq.~\eqref{eq: classical fidelity bound}. We evaluate both of these scenarios paired with postselection of the $\ket{\Phi^+}_{X_{k+1}A_{k}}$ outcome.

After $k$ successful BSMs, the PPT negativity of the reduced state $\rho_{A_kC_1}$ of qubits $A_kC_1$ is given by
\begin{equation}
    \mathcal N^{k,M}_{A_k:C_1} =\max\left\{0\,,\, \frac{U - V + \sqrt{V^2+W^2}}{2UM}\right\}\,,
\end{equation}
where $U=(k+2)(k+3)$, $V=2(k+1)(M+2)$ and $W=k(k+3)$. The above equation is found by tracing out all other qubits using Dicke bases, and then applying Eq.~\eqref{eq: negativity definition} to the reduced density matrix. If we now subject the state to a global depolarizing noise, such that 
\begin{equation}\label{eq: depo noise model}
    \rho_{A_kC_1}^\prime(\mu)=(1-\mu)\rho_{A_kC_1} + \frac{\mu}{4}\Id
\end{equation}
with $\mu\in[0,1]$, the negativity function reaches a value of zero at
\begin{equation}
    \Tilde{\mu}_{A_kC_1}(k,M) = \frac{2(U-V+\sqrt{V^2+W^2})}{(k^2+k+2)(M+2)+2\sqrt{V^2+W^2}}\,.
\end{equation}
Therefore, $\rho_{A_kC_1}^\prime(\mu)$ is entangled for $\mu<\Tilde{\mu}_{A_kC_1}(k,M)$.

Returning to the fidelity thresholds, we first examine the increase of fidelity from the previous iteration's ideal value of $\gamma(k,M)$.
\begin{theorem}\label{theorem: A_kC_1 entanglement increasing fidelity not necessary}
    For sufficiently large $M$, two-qubit entanglement in $A_{k}:C_i$ is not necessary for increasing the postselected fidelity for receiver $C_i$ when performing the $(k+1)$th BSM on $X_{k+1}A_k$ in the $N\to M$ telecloning protocol.
\end{theorem}
\begin{figure}
    \centering
    \includegraphics[width=\linewidth]{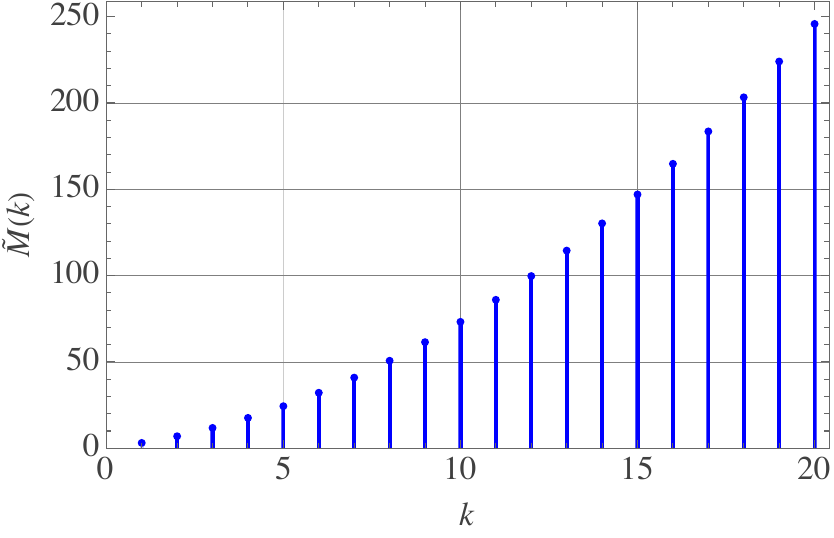}
    \caption{After $k$ prior successful BSMs, we plot the threshold value of number of receivers $\Tilde{M}(k)$ beyond which pairwise entanglement in $A_{k}:C_i$ is no longer necessary for increasing the $i$th receiver's postselected fidelity beyond $\gamma(k,M)$, for the $(k+1)$th BSM. Since this plot is based on just one simple noise model, these values are the upper bounds of this threshold value of $M$, with other noise models potentially bringing this value down for each $k$.}
    \label{fig: tilde M vs k}
\end{figure}

\textbf{Proof.} One way to probe the two-qubit entanglement in $A_{k}:C_1$ is by finding the noise strength beyond which the fidelity cannot be improved on $\gamma(k,M)$. In essence, we ask for which value of $\mu$ is
\begin{equation}
    \gamma(k,M) = \frac{(1-\mu)p_{k+1}\gamma(k+1,M)+\frac{\mu}{2}\cdot\frac{1}{4}}{(1-\mu)p_{k+1}+\frac{\mu}{4}}
\end{equation}
satisfied, since the probability of outcome $\ket{\Phi^+}_{X_{k+1}A_k}$ using state $\Id/4$ is $1/4$, with a resulting teleportation fidelity of $1/2$. Solving this equation for $\mu$ gives us a value of
\begin{equation}
    \Tilde{\mu}_\gamma = \frac{4 (k+1)}{k^3+5k^2+10k+4}\,.
\end{equation}
To test for necessity of two-qubit entanglement, we search for the conditions that satisfy the inequality $\Tilde{\mu}_{A_kC_1}< \Tilde{\mu}_\gamma$. This gives us the conditions where the fidelity may be increased from $\gamma(k,M)$ despite the state of $A_k:C_1$ being separable. Under the depolarizing noise described in Eq.~\eqref{eq: depo noise model}, this happens if and only if
\begin{equation}
    M > \frac{k\left(U - 4 + \sqrt{U^{2}-4(3k+5)}\right)}{2(2k+1)} = \Tilde{M}(k)\,.
\end{equation}
If $M>\Tilde{M}(k)$, then there must always exist a range of $\mu$ where the fidelity may increase despite the state being separable, showing that entanglement is not necessary in these conditions. A simple example is for the second BSM, where $\Tilde{M}(1)=\frac{2}{3} \left(2+\sqrt{7}\right)\approx 3.09717$. This means that, for an $N\to 4$ (or higher) protocol, no two-qubit entanglement is necessary for increasing the fidelity for a given receiver for the second BSM. We display the function $\Tilde{M}(k)$ for up to $k=20$ in Fig.~\ref{fig: tilde M vs k}, which suggests that the reliance on pairwise entanglement decreases as $M$ grows.
Since the permutational invariance of the receivers means that we do not lose any generality for selecting $C_1$, and $M$ has no upper bound, we thus prove Theorem~\ref{theorem: A_kC_1 entanglement increasing fidelity not necessary} by explicit example. $\square$

We now move on to outperforming the classical communication limit.
\begin{theorem}
    Two-qubit entanglement in $A_{k}:C_i$ is not necessary for achieving a postselected teleportation fidelity surpassing the classical communication limit for receiver $C_i$, when performing the $(k+1)$th operation on $X_{k+1}A_k$ in an $N\to M$ telecloning protocol, for any $1\leq k\leq M-1$.
\end{theorem}

\textbf{Proof.} In a fashion akin to the proof of Theorem~\ref{theorem: A_kC_1 entanglement increasing fidelity not necessary}, we seek to solve
\begin{equation}
    \mathcal{F}_{\mathrm{cl}}(k+1) =  \frac{(1-\mu)p_{k+1}\gamma(k+1,M)+\frac{\mu}{2}\cdot\frac{1}{4}}{(1-\mu)p_{k+1}+\frac{\mu}{4}}\,,
\end{equation}
giving the solution
\begin{equation}
    \Tilde{\mu}_{\mathrm{c}} = \frac{4 (k+1)}{k (M+4)+2 (M+2)}\,.
\end{equation}
It can be shown that $\Tilde{\mu}_{A_kC_1}<\Tilde{\mu}_{\mathrm{c}}$ for all valid $k$ and $M$. Therefore, there always exists a range of $\mu$ where we have a separable state, yet the classical-communication limit is surpassed. This implies that two-qubit entanglement is not necessary here, in contrast to the $1\to M$ case where it is necessary. $\square$
 
\section{Conclusions}\label{sec: conclusions}
In this manuscript, we introduced a novel, experimentally accessible, symmetric $N\to M$ telecloning protocol for arbitrary unknown qubit states. Although probabilistic, a successful protocol achieves optimal cloning fidelity via postselection of certain BSM outcomes, with the success probability of each individual BSM increasing throughout the protocol. Furthermore, we highlighted the fact that the average fidelity of a given BSM is equal to the optimal fidelity of the previous BSM, and thus the protocol remains structurally consistent as $N$ increases. We have shown that unsuccessful protocols can be considered state-dependent symmetric $N\to M$ telecloning protocols, with Alice and all receivers being able to infer information about the unknown state. The amount of information they learn in these scenarios increases with the number of successful measurements alongside a single unsuccessful measurement, and also depends on the probability distribution of input states. In particular, this information gain heavily favors distributions based on the Pauli eigenstates.

Additionally, we assessed how the entanglement of the resource state under the global depolarization and local dephasing noise affects the teleportation fidelity for each receiver. We found that, for the standard $1\to M$ protocol, it is the pairwise entanglement between the port qubit and each receiver that is crucial to surpass the classical-communication bound on fidelity, despite the resource state possessing GME and global entanglement. In fact, we have shown that these multipartite entanglement classes are not even necessary for optimal performance of the protocol, and have further provided evidence to suggest that global entanglement may not be sufficient either. For our $N\to M$ protocol, we evaluated the entanglement-fidelity dynamics for two different benchmarks. The first was to increase the fidelity from one BSM to the next, and exemplified that pairwise entanglement is not generally necessary for this feat. The second was to outperform the optimal fidelity when limited to classical communication; we explicitly showed that entanglement is not necessary for this when $N\geq 2$.

\section{Outlook}\label{sec: outlook}
We now outline several future research directions.

Perhaps the biggest drawback of this protocol is its probabilistic nature. While the probability does indeed increase with $k$, the overall probability will still decrease with $N$ since these probabilities compound. Although our results suggest that this total probability is hence determined by $N$, we hypothesize that it is instead the number of measurements made. We therefore think an interesting line of study would be to consider more generalized sequential/parallel measurements (such as the global POVMs in Ref.~\cite{Dur_1999}) made on groups of copies.
%For example, if we have $N=4$ copies in system $X_1X_2X_3X_4$, we could consider making joint measurements on pairs of copies and auxiliaries, such as $X_1X_2PA_1$ and $X_3X_4A_2A_3$. This first measurement could indeed be deterministic when paired with LOCC, with any subsequent measurements perhaps being probabilistic. In fact, these groups of copies may not necessarily need to be the same size; for our example, there is no reason, in principle, why we could not perform a joint measurement on $X_1X_2X_3PA_1A_2$ before making a Bell-type measurement on $X_4A_3$.
Such a line of inquiry could lead to a trade-off relation between success probability and experimental simplicity, as well as offering deeper insight into the relationship with remote information gain.

Other natural extensions to our protocol (inspired by established results for the $1\to M$ case) include developing an \textit{asymmetric} $N\to M$ telecloning protocol~\cite{Ferraro_2005, Chen_2007}, giving more freedom to how the quantum information may be spread among the receivers. Either symmetric or asymmetric $N\to M$ protocols could also be generalized to $d$-dimensional (or even continuous-variable) systems~\cite{Araneda_2016, Zhang_2008_cont, vanLoock_2001,Ghiu_2003}. Moreover, since our protocol has only been explicitly shown for pure states, one may consider its suitability for unknown mixed states~\cite{Fanizza_2026}.

Finally, further investigations could be conducted into the necessary and/or sufficient resources for $N\to M$ telecloning. The fact that pairwise entanglement is generally not necessary for $N\geq 2$ suggests that there may be a different resource at play to allow for quantum advantage. Perhaps a contender for this could be a more general form of quantum correlations, such as the quantum discord~\cite{Ollivier,Henderson,Modi}.
In any case, the necessity of entanglement for $1\to M$ telecloning means that it may, in principle, be used for entanglement verification on quantum networks.

\acknowledgments
A.G.H. thanks Mauro Paternostro and Giorgio Zicari for their guidance throughout the early stages of this project, as well as Mio Murao for insightful discussions. This work was supported by the Institute of Information \& Communications Technology Planning \& Evaluation (IITP) grant, funded by the Korea government (MSIT) (No. 2022-0-00463), the Department for the Economy Northern Ireland under the US-Ireland R\&D Partnership Programme, and Munster Technological University’s TU RISE Research to Impact co-funded by the Government of Ireland and the European Union through the ERDF Southern, Eastern $\&$ Midland Regional Programme 2021-27. H.K. is supported by KIAS individual grant number CG085302 at the Korea Institute for Advanced Study.

\section*{Statement on use of Generative AI}
The authors confirm that various models of Claude were used to assist with the formulation of some of the proofs, particularly for finding relevant 
combinatorial identities. Meanwhile, all semi-colons, em-dashes, and Oxford commas in this manuscript were human-generated.
 
\bibliography{references.bib}
\begin{widetext}
\appendix
\section{Derivation of $\ket{\psi_k}_{A^\prime C}$}\label{app: derivation of state}
Here, we justify Eq.~\eqref{eq: normalized state with general probablity}. For the sake of convenience, we again define the state after $k$ successful BSMs as
\begin{equation}
\begin{split}
    \ket{\psi_k}_{A^\prime C}=&\left(\prod_{i=1}^{k} \frac{1}{\sqrt{p_i}}\right) \sum_{t=0}^{k}\sum_{j=0}^{M-t} \sqrt{\frac{(M-j)_{t}^\downarrow (j)_{k-t}^\downarrow}{2^{k}(M+1)_{k+1}^\downarrow}}
    % &\ \ \ \ \ \ \ \times\begin{pmatrix}
    %     k\\
    %     t
    % \end{pmatrix}
    % \alpha^{k-t}\beta^t\ket{M-t-j}_{A^\prime}^{M-k}\ket{M-j}_C^M\,.
    \begin{pmatrix}
        k\\
        t
    \end{pmatrix}
    \alpha^{k-t}\beta^t\big|{D_{M-t-j}^{M-k}}\big\rangle_{A^\prime}\ket{D_{M-j}^{M}}_C\,.
\end{split}
\end{equation}
To aid understanding of the origin of each term, we start by considering the case of the second BSM (i.e., $k=2$). Referring back to Eq.~\eqref{eq: 1 to M transformed state}, the joint state of $X_2$ and $AC$ (before the second BSM) is
\begin{equation}
\begin{split}
    &\ket{\varphi}_{X_2}\otimes \ket{\psi_1}_{AC}=(\alpha\ket{0}+\beta\ket{1})_{X_2} \otimes \sum_{j=0}^{M-1} \underbrace{\sqrt{\frac{2(M-j)}{(M+1)M}}}_{c_j}\Big[\alpha 
    % \ket{j}^{M-1}_A \ket{j}^M_C + \beta \ket{M-1-j}^{M-1}_A \ket{M-j}_C^M
    \ket{D_j^{M-1}}_A \ket{D_j^M}_C + \beta \big|{D_{M-1-j}^{M-1}}\big\rangle_A \ket{D_{M-j}^M}_C
    \Big]\,.
\end{split}
\end{equation}
Now, to find the impact of outcome $\ket{\Phi^+}_{X_2A_1}$ on the above state, we consider factoring out the states $\ket{00}_{X_2A_1}$ and $\ket{11}_{X_2A_1}$. For the former, we have
\begin{equation}\label{eq: 00}
\begin{split}
    &\ket{00}_{X_2A_1}\otimes c_j \Bigg[\alpha^2 \begin{pmatrix}
        M-1\\
        j
    \end{pmatrix}
    \begin{pmatrix}
        M-2\\
        j
    \end{pmatrix}^{\frac{1}{2}}
    % \ket{j}_{A^\prime}^{M-2}\ket{j}_C^M
    \ket{D_j^{M-2}}_{A^\prime}\ket{D_j^M}_C+\alpha\beta \begin{pmatrix}
        M-1\\
        j
    \end{pmatrix}^{-\frac{1}{2}}
    \begin{pmatrix}
        M-2\\
        j-1
    \end{pmatrix}^{\frac{1}{2}}
    % \ket{M-1-j}_{A^\prime}^{M-2}\ket{M-j}_C^M
    \big|{D_{M-1-j}^{M-2}}\big\rangle_{A^\prime}\ket{D_{M-j}^M}_C
    \Bigg]
\end{split}
\end{equation}
for a given value of $j$.
Similarly, we have
\begin{equation}\label{eq: 11}
    \begin{split}
        &\ket{11}_{X_2A_1} \otimes c_j \Bigg[\alpha\beta \begin{pmatrix}
        M-1\\
        j
    \end{pmatrix}^{-\frac{1}{2}}
    \begin{pmatrix}
        M-2\\
        j-1
    \end{pmatrix}^{\frac{1}{2}}
    % \!\ket{j-1}_{A^\prime}^{M-2}\!\ket{j}_C^M\!+
    \ket{D_{j-1}^{M-2}}_{A^\prime}\ket{D_j^M}_C\beta^2\begin{pmatrix}
        M-1\\
        j
    \end{pmatrix}^{-\frac{1}{2}}\!
    \begin{pmatrix}
        M-2\\
        j
    \end{pmatrix}^{\frac{1}{2}}
    % \!\ket{M-2-j}_{A^\prime}^{M-2}\!\ket{M-j}_C^M
    \big|{D_{M-2-j}^{M-2}}\big\rangle_{A^\prime}\!\ket{D_{M-j}^M}_C
    \Bigg]\,.
    \end{split}
\end{equation}
We then notice the identities
\begin{align}\label{eq: identity 1}
    c_j\begin{pmatrix}
        M-1\\
        j
    \end{pmatrix}^{-\frac{1}{2}}\!
    \begin{pmatrix}
        M-2\\
        j
    \end{pmatrix}^{\frac{1}{2}}&= \sqrt{\frac{2(M-j)(M-1-j)}{(M+1)M(M-1)}}\\\label{eq: identity 2}
    c_j\begin{pmatrix}
        M-1\\
        j
    \end{pmatrix}^{-\frac{1}{2}}\!
    \begin{pmatrix}
        M-2\\
        j-1
    \end{pmatrix}^{\frac{1}{2}}&= \sqrt{\frac{2(M-j)j}{(M+1)M(M-1)}}\,,
\end{align}
allowing us to factor out the state $\ket{\Phi^+}_{X_2A_1}$ as
\begin{equation}\label{eq: phi+ state unsimplified}
    \begin{split}
        &\ket{\Phi^+}_{X_2A_1} \otimes\Bigg\{ \sum_{j=0}^{M-2} \Bigg[ \sqrt{\frac{(M-j)(M-1-j)}{M(M+1)(M-1)}}
        % \left(\alpha^2\ket{j}^{M-2}_{A^\prime} \ket{j}^M_C + \beta^2\ket{M-2-j}_{A^\prime}^{M-2}\ket{M-j}_C^M\right)
        \left(\alpha^2\ket{D_j^{M-2}}_{A^\prime} \ket{D_j^M}_C + \beta^2\big|D_{M-2-j}^{M-2}\big\rangle_{A^\prime}\ket{D_{M-j}^M}_C\right)
        \Bigg]\\
    &+\sum_{j=0}^{M-2}\alpha\beta \sqrt{\frac{(M-j)j}{M(M+1)(M-1)}}
    % \ket{M-1-j}^{M-2}_{A^\prime} \ket{M-j}^M_C
    \big|D_{M-1-j}^{M-2}\big\rangle_{A^\prime} \ket{D_{M-j}^M}_C
    +\sum_{j=0}^{M-2} \alpha\beta \sqrt{\frac{(M-j)j}{M(M+1)(M-1)}}
    % \ket{j-1}^{M-2}_{A^\prime} \ket{j}_C^M
    \ket{D_{j-1}^{M-2}}_{A^\prime} \ket{D_j^M}_C
    \Bigg\}\,,
    \end{split}
\end{equation}
where the $j=M-1$ terms vanish. In the second line, we may write the two sums as a single one, since taking $j\to M-j$ gives equivalent summation limits and results in the same sum of states. Also, applying this logic to the $\alpha^2$ terms, we see that
\begin{align}
\begin{split}
    \sum_{j=0}^{M-2}\alpha\beta \sqrt{\frac{(M-j)j}{M(M+1)(M-1)}}
    % \ket{M-1-j}^{M-2}_{A^\prime} \!\ket{M-j}^M_C
    \big|{D_{M-1-j}^{M-2}}\big\rangle_{A^\prime} \!\ket{D_{M-j}^M}_C
    &\equiv\sum_{j=0}^{M-2} \alpha\beta \sqrt{\frac{(M-j)j}{M(M+1)(M-1)}}
    % \ket{j-1}^{M-2}_{A^\prime} \ket{j}_C^M
    \ket{D_{j-1}^{M-2}}_{A^\prime} \ket{D_j^M}_C
\end{split}\\
\begin{split}
    \sum_{j=0}^{M-2} \sqrt{\frac{(M-j)(M-1-j)}{M(M+1)(M-1)}}\alpha^2
    % \ket{j}^{M-2}_{A^\prime} \ket{j}^M_C
    \ket{D_j^{M-2}}_{A^\prime} \ket{D_j^M}_C
    &\equiv \sum_{j=0}^{M-2} \sqrt{\frac{j(j-1)}{M(M+1)(M-1)}}\alpha^2
    % \ket{M-j}^{M-2}_{A^\prime} \ket{M-j}^M_C
    \big|{D_{M-j}^{M-2}}\big\rangle_{A^\prime} \ket{D_{M-j}^M}_C\,.
\end{split}
\end{align}
Thus, for success probability $p_2$, the normalized state becomes
\begin{equation}
\begin{split}
    \ket{\psi_2}_{A^\prime C}=&\left(\prod_{i=1}^{2} \frac{1}{\sqrt{p_i}}\right) \sum_{t=0}^{2}\sum_{j=0}^{M-t} \sqrt{\frac{(M-j)_{t}^\downarrow (j)_{2-t}^\downarrow}{2^{2}(M+1)_{3}^\downarrow}}
    \begin{pmatrix}
        2\\
        t
    \end{pmatrix}
    % \alpha^{2-t}\beta^t\ket{M-t-j}_{A^\prime}^{M-2}\ket{M-j}_C^M
    \alpha^{2-t}\beta^t\big|{D_{M-t-j}^{M-2}}\big\rangle_{A^\prime}\ket{D_{M-j}^M}_C\,,
\end{split}
\end{equation}
where the factor of $\sqrt{1/2^2}$ comes from $p_1=1/4$, as seen in Eq.~\eqref{eq: XPAC pure state}.

The above logic may be applied generally for $k$ successful BSMs. Any states with the same $\alpha^{k-t}\beta^t$ coefficients will have had $k-t$ zeros and $t$ ones `removed', also giving the same falling factorials of $\sqrt{(j)_{k-t}^\downarrow}$ and $\sqrt{(M-j)_t^\downarrow}$ respectively. This, as well as the $(M+1)_{k+1}^\downarrow$ term in the denominator, is a result of the identities in Eqs.~\eqref{eq: identity 1} and \eqref{eq: identity 2} being generalized to $k$ successful BSMs, with the $k$th BSM contributing the factor $\sqrt{1/(M+1-k)}$. Meanwhile, the $2^k$ term comes from the equal superposition of states present in $\ket{\Phi^+}$ each time a successful BSM is made, an example of this being the canceling of the factor of 2 from the numerator in Eq.~\eqref{eq: phi+ state unsimplified}. Similar logic may be used to find the other states shown in Eqs.~\eqref{eq: Phi pm states AC} and \eqref{eq: Psi pm states AC}. See Appendix~\ref{app: multiple outcomes} for more detail on how the binomial coefficient is transformed by each BSM outcome.

\section{Proof of Theorem~\ref{theorem: probability}}\label{app: proof of theorem 1}
Let us start by not assuming the form of the $i$th BSM success probability, which we here denote as $\Tilde{p}_i$. Rather, we wish to prove that $\Tilde{p}_i\equiv p_i$. Using the state in Eq.~\eqref{eq: normalized state with general probablity}, we can work out the probability $\Tilde{p}_{k+1}$ provided we have knowledge of all $\{\Tilde{p}_1,\cdots,\Tilde{p}_k\}$ that come before it (with $1\leq k\leq N-1$). From Eq.~\eqref{eq: XPAC pure state}, we know that $\Tilde{p}_1=\frac{1}{4}=p_1$, and so we have enough information to obtain $\Tilde{p}_{k+1}$ through recursion. We therefore have
\begin{equation}\label{eq: tilde p k+1}
\begin{split}
    \Tilde{p}_{k+1} = \left(\prod_{i=1}^k \frac{1}{\Tilde{p}_i}\right) &\sum_{t=0}^{k+1}\sum_{j=0}^{M-t} \frac{(M-j)_{t}^\downarrow (j)_{k-t+1}^\downarrow}{2^{k+1}(M+1)_{k+2}^\downarrow}\begin{pmatrix}
        k+1\\
        t
    \end{pmatrix}^2
    \abs{\alpha}^{2(k+1-t)} \abs{\beta}^{2t}\,.
\end{split}
\end{equation}
The $j$-dependent terms can be equivalently written as
\begin{equation}
    \sum_{j=0}^{M-t}(M-j)_{t}^\downarrow (j)_{k-t+1}^\downarrow = T_{k,t}\sum_{j=0}^{M-t}
    \begin{pmatrix}
        M-j\\
        t
    \end{pmatrix}
    \begin{pmatrix}
        j\\
        k-t+1
    \end{pmatrix}\,,
\end{equation}
where $T_{k,t}=(k-t+1)!t!$. This can then be further simplified by completely eliminating the $j$-dependence as
\begin{equation}
    \sum_{j=0}^{M-t}(M-j)_{t}^\downarrow (j)_{k-t+1}^\downarrow= T_{k,t}
    \begin{pmatrix}
        M+1\\
        k+2
    \end{pmatrix}\,.
\end{equation}
Here, we have used the \textit{upside-down Chu--Vandermonde identity}~\cite{Riordan_1979,Grinberg_2022,Krapf_2026}, which is stated, for $a,b,c\in\mathbb{N}$, as
\begin{equation}\label{eq: upside down chu vandermonde}
    \sum_{j=0}^a
    \begin{pmatrix}
        j\\
        b
    \end{pmatrix}
    \begin{pmatrix}
        a-j\\
        c
    \end{pmatrix}
    \equiv
    \begin{pmatrix}
        a+1\\
    b+c+1\end{pmatrix}\,.
\end{equation}
This identity is appropriate here since all terms for $j=M-t+1$ through $j=M$ are zero. By evaluating and simplifying the rest of the terms, Eq.~\eqref{eq: tilde p k+1} can be written as
\begin{equation}
\begin{split}
    \Tilde{p}_{k+1} &= \left(\prod_{i=1}^k \frac{1}{\Tilde{p}_i}\right)\frac{(k+1)!\begin{pmatrix}
        M+1\\
        k+2
    \end{pmatrix}}{2^{k+1}(M+1)_{k+2}^\downarrow} \sum_{t=0}^{k+1} \begin{pmatrix}
        k+1\\
        t
    \end{pmatrix}
    \abs{\alpha}^{2(k+1-t)} \abs{\beta}^{2t}\\
    &= \left(\prod_{i=1}^k \frac{1}{\Tilde{p}_i}\right)\frac{1}{2^{k+1}(k+2)}\,,
\end{split}
\end{equation}
where, in the last line, we have used the binomial theorem and the normalization constraint with
\begin{equation}
    \sum_{t=0}^{k+1} \begin{pmatrix}
        k+1\\
        t
    \end{pmatrix}
    \abs{\alpha}^{2(k+1-t)} \abs{\beta}^{2t} \equiv (\abs{\alpha}^{2}+ \abs{\beta}^{2})^{k+1} = 1\,.
\end{equation}
Finally, we substitute $\Tilde{p}_i=p_i$, where $\prod_{i=1}^k (1/p_i)=2^k(k+1)$, and simplify to find that
\begin{equation}
    \Tilde{p}_{k+1} = \frac{k+1}{2(k+2)} = p_{k+1}\,.
\end{equation}
Since $p_i$ satisfies both the base case and the recursion, it must therefore be the probability of the $i$th BSM being successful. The probability is clearly not a function of $M$, thus concluding the proof. $\square$

\section{Proof of Lemma~\ref{lemma: covariance}}\label{app: proof of lemma 2}
Let us start by showing that the $1\to M$ protocol is covariant to serve as a base case for our proof. Recall the \textit{ricochet property} (sometimes called the \textit{transpose trick})~\cite{Wilde,Larocca_2022} of a $d$-dimensional maximally entangled bipartite state $\ket{\Phi}_{AB}$ as
\begin{equation}
    (K_A\otimes \Id_B) \ket{\Phi}_{AB} \equiv (\Id_A \otimes K_B^T) \ket{\Phi}_{AB}\,,
\end{equation}
where
\begin{equation}
    \ket{\Phi}_{AB} = \frac{1}{\sqrt{d}}\sum_{i=0}^{d-1}\ket{i}_A\ket{i}_B
\end{equation}
and $K$ is any linear operator. From this identity, we can further establish that
\begin{equation}
    (V_A^* \otimes V_B)\ket{\Phi}_{AB} = (V_A^TV_A^* \otimes\, V_B^\dagger V_B)\ket{\Phi}_{AB} = \ket{\Phi}_{AB}\,,
\end{equation}
for any unitary operator $V$. Now, since the telecloning resource state $\ket{\psi_{\mathrm{T}}}$ can be written as an $(M+1)$-dimensional maximally entangled state between systems $PA$ and $C$, as shown in Eq.~\eqref{eq: telecloning resource state}, we thus have the identity
\begin{equation}\label{eq: telecloning state invariance}
    \left(V^{* \otimes M}_{PA} \otimes V_{C}^{\otimes M}\right)\ket{\psi_{\mathrm{T}}} = \ket{\psi_{\mathrm{T}}}\,.
\end{equation}
For the $1\to M$ protocol with a rotated input, we have
\begin{equation}\label{eq: 1 to M covariance}
\begin{split}
    \bra{\Phi^+}_{XP}U_X\ket{\varphi}_X\ket{\psi_{\mathrm{T}}}&=\bra{\Phi^+}_{XP}U^T_P\ket{\varphi}_X\ket{\psi_{\mathrm{T}}}\\
    &= \left(U^{*\otimes M-1}_{A}\otimes U_C^{\otimes M}\right)\bra{\Phi^+}_{XP}\ket{\varphi}_X\ket{\psi_{\mathrm{T}}}\,,
\end{split}
\end{equation}
where, in the last line, we use Eq.~\eqref{eq: telecloning state invariance}. We have assumed postselection of $\ket{\Phi^+}_{XP}$, but it can be similarly shown for the other BSM outcomes after application of the appropriate LOCC outlined in Table~\ref{tab: telecloning}. It is implied by Eq.~\eqref{eq: 1 to M covariance} that the $1\to M$ protocol is indeed covariant.

Now, suppose Alice has already had $k$ BSMs with each yielding outcome $\ket{\Phi^+}$. Let us assume that, after $k$ BSMs, the state of the receivers and the remaining auxiliary system can be written as
\begin{equation}\label{eq: k assumption of state}
    \ket{\psi_k(U\ket{\varphi})}_{A^\prime C} = \left(U_{A_k\cdots A_{M-1}}^{*\otimes M-k} \otimes U_C^{\otimes M}\right)\ket{\psi_k(\ket{\varphi})}_{A^\prime C}\,.
\end{equation}
Say now that we perform BSM number $k+1$, postselecting on outcome $\ket{\Phi^+}_{X_{k+1}A_k}$. The unnormalized state will be
\begin{equation}
\begin{split}
    \ket{\Tilde{\psi}_{k+1}(U\ket{\varphi})}_{A^\prime C} &= \bra{\Phi^+}_{X_{k+1}A_k} \ket{\varphi^\prime}_{X_{k+1}} \ket{\psi_{k}(U\ket{\varphi})}_{A^\prime C}  \\
    &=\Tilde{U} \bra{\Phi^+}_{X_{k+1}A_k} \left(\alpha^\prime\ket{0}+\beta^\prime\ket{1}\right)_{X_{k+1}} U_{A_k}^{*}\ket{\psi_k(\ket{\varphi})}_{A^\prime C} 
\end{split}
\end{equation}
where $\ket{\varphi^\prime}=U\ket{\varphi} = \alpha^\prime\ket{0} + \beta^\prime\ket{1}$, $\Tilde{U}=U_{A_{k+1}\cdots A_{M-1}}^{*\otimes M-1-k} \otimes U_C^{\otimes M}$. Noticing that
\begin{equation}
\begin{split}
    \bra{\Phi^+}_{X_{k+1}A_k} \ket{\varphi^\prime}_{X_{k+1}} &= \frac{1}{\sqrt{2}}(\alpha^\prime\bra{0} + \beta^\prime \bra{1})_{A_k}\\
    &=\frac{1}{\sqrt{2}}\bra{\varphi^{\prime*}}_{A_k}\,,
\end{split}
\end{equation}
we can write the state as
\begin{equation}
    \ket{\Tilde{\psi}_{k+1}(U\ket{\varphi})}_{A^\prime C} = \frac{1}{\sqrt{2}} \Tilde{U}  \bra{\varphi^{\prime*}}_{A_k} U_{A_k}^{*}\ket{\psi_k(\ket{\varphi})}_{A^\prime C}  \,.
\end{equation}
We further notice that
\begin{equation}
    \alpha^\prime\bra{0} + \beta^\prime\bra{1} = (\alpha\bra{0} + \beta\bra{1})U^T \,,
\end{equation}
and so
\begin{equation}
    \left(\alpha^\prime\bra{0}+\beta^\prime\bra{1}\right)_{A_{k}} U_{A_k}^{*} = \alpha\bra{0}_{A_{k}} + \beta\bra{1}_{A_{k}}\,.
\end{equation}
We may therefore finally write the normalized state
\begin{equation}
\begin{split}
    \ket{\psi_{k+1}(U\ket{\varphi})}_{A^\prime C} &= \frac{\mathcal{G}}{\sqrt{2}} \left(\alpha\bra{0}+\beta\bra{1}\right)_{A_{k}} \Tilde{U} \ket{\psi_k(\ket{\varphi})}_{A^\prime C}  \\
    &=   \left(U_{A_{k+1}\cdots A_{M-1}}^{*\otimes M-1-k} \otimes U_C^{\otimes M}\right) \ket{\psi_{k+1}(\ket{\varphi})}_{A^\prime C}\,,
\end{split}
\end{equation}
where $\mathcal{G}=\sqrt{1/p_{k+1}}$. Since the assertion of Eq.~\eqref{eq: k assumption of state} is true for $k=1$, by induction, the $N\to M$ telecloning protocol is thus covariant for all integers $1\leq N \leq M$. This concludes the proof. $\square$

\section{Multiple unsuccessful outcomes}\label{app: multiple outcomes}
We can see in Eqs.~\eqref{eq: Phi pm states AC} and \eqref{eq: Psi pm states AC} that the outcome of a BSM does not affect the coefficients $\mathcal{B}_{k,M}^{j,t}$ associated with each state. That is to say, all outcomes transform these values via $k\to k+1$ for the $(k+1)$th measurement. What the specific outcome instead determines, is the amplitude of each state via the weighting of the $\alpha^{k-t}$ and $\beta^t$ coefficients by sums of binomials (as well as the overall normalization factor $\mathcal{A}_{k,M}^\Pi$). As briefly mentioned in Sec.~\ref{sec: performances of unsuccessful outcomes}, the reason why $\ket{\Phi^+}_{X_{k+1}A_{k}}$ is so crucial for the $(k+1)$th BSM is that it transforms the weighting, for a given value of $t$, as
\begin{equation}
\begin{split}
    &\alpha^{k-t}\beta^t\begin{pmatrix}
        k\\
        t
    \end{pmatrix} \xrightarrow{\Phi^+} \alpha^{k+1-t}\beta^t\begin{pmatrix}
        k+1\\
        t
    \end{pmatrix}\,.
\end{split}
\end{equation}
In fact, if we have $N=K+1$ total measurements comprising $K$ successful outcomes and a single unsuccessful outcome, the state is actually independent of the value of $k$ for which the unsuccessful measurement occurred. The final state (and overall probability) will be the same.

We now investigate how these coefficients transform in the event of multiple unsuccessful BSM outcomes.
In the event of outcome $\ket{\Phi^-}_{X_{K+1}A_K}$ (following $K$ previous successful outcomes), the coefficient for index $t$ instead transforms as
\begin{equation}
    \begin{split}
        \alpha^{K-t}\beta^t\begin{pmatrix}
        K\\
        t
    \end{pmatrix}
    \xrightarrow{\Phi^-} \alpha^{K+1-t}\beta^t\Bigg[\begin{pmatrix}
        K\\
        t
    \end{pmatrix} -\begin{pmatrix}
        K\\
        t-1
    \end{pmatrix}\Bigg]\,.
    \end{split}
\end{equation}
If we then perform another BSM and get the outcome $\ket{\Phi^-}_{X_{K+2}A_{K+1}}$, then we have
\begin{equation}
    \begin{split}
        &\alpha^{K+1-t}\beta^t\Bigg[\begin{pmatrix}
        K\\
        t
    \end{pmatrix} -\begin{pmatrix}
        K\\
        t-1
    \end{pmatrix}\Bigg]\xrightarrow{\Phi^-} \alpha^{K+2-t}\beta^t \Bigg[\begin{pmatrix}
        K\\
        t
    \end{pmatrix} - 2\begin{pmatrix}
        K\\
        t-1
    \end{pmatrix} + \begin{pmatrix}
        K\\
        t-2
    \end{pmatrix}\Bigg]\,.
    \end{split}
\end{equation}
No generality is lost here, since the transformations induced by a successful measurement commute with any number of transformations resulting from a singular type of unsuccessful outcome. That is to say, the unsuccessful outcomes do not generally commute with each other, but $\nu$ unsuccessful outcomes of the same kind (such as $\ket{\Phi^-}$) do commute with $K$ successful outcomes, for any $\nu+K=N\leq M$.

We find that each time we obtain the outcome $\ket{\Phi^-}$, the square brackets undergo a \textit{backward difference operation} $\nabla$, where we define
\begin{equation}
    \nabla f(t) = f(t) - f(t-1)
\end{equation}
for a function $f$. In our scenario, we therefore have
\begin{equation}
    \nabla^\nu \begin{pmatrix}
        K\\
        t
    \end{pmatrix}\,,
\end{equation}
where
\begin{equation}
\begin{split}
    \nabla^0 \begin{pmatrix}
        K\\
        t
    \end{pmatrix} = \begin{pmatrix}
        K\\
        t
    \end{pmatrix}\ \ ,\ \ 
    \nabla^1 \begin{pmatrix}
        K\\
        t
    \end{pmatrix} = \begin{pmatrix}
        K\\
        t
    \end{pmatrix} - \begin{pmatrix}
        K\\
        t-1
    \end{pmatrix}\,.
\end{split}
\end{equation}
We may then write $\nabla = \Id - E^{-1}$, where $\Id$ is the identity and $E^{-1}f(t)=f(t-1)$. Since $\Id$ and $E^{-1}$ commute, this can be expanded by the binomial theorem as
\begin{equation}
    \nabla^\nu f(t) = (\Id - E^{-1})^\nu f(t) = \sum_{i=0}^\nu \begin{pmatrix}
        \nu\\
        i
    \end{pmatrix}
    (-1)^i (E^{-1})^i f(t)\,,
\end{equation}
where $(E^{-1})^i f(t) = f(t-i)$. We therefore have
\begin{equation}
    \nabla^\nu \begin{pmatrix}
        K\\
        t
    \end{pmatrix} = \sum_{i=0}^\nu (-1)^i \begin{pmatrix}
        \nu\\
        i
    \end{pmatrix}
     \begin{pmatrix}
        K\\
        t-i
    \end{pmatrix}\,.
\end{equation}

Following similar analysis of the $\ket{\Psi^\pm}$ outcomes, we arrive at the operators
\begin{align}
    \nabla^\nu_{\Phi^+}\! \begin{pmatrix}
        K\\
        t
    \end{pmatrix}\alpha^{K-t}\beta^t &=
     \begin{pmatrix}
        K+\nu\\
        t
    \end{pmatrix}
    \alpha^{K+\nu-t}\beta^t\\
    \nabla^\nu_{\Phi^-}\! \begin{pmatrix}
        K\\
        t
    \end{pmatrix}\alpha^{K-t}\beta^t &= \alpha^{K+\nu-t}\beta^t\sum_{i=0}^\nu \varepsilon_i \begin{pmatrix}
        \nu\\
        i
    \end{pmatrix}
     \begin{pmatrix}
        K\\
        t-i
    \end{pmatrix}\\
    \nabla^\nu_{\Psi^+} \begin{pmatrix}
        K\\
        t
    \end{pmatrix}\alpha^{K-t}\beta^t &= \sum_{i=0}^\nu  \begin{pmatrix}
        \nu\\
        i
    \end{pmatrix}
     \begin{pmatrix}
        K\\
        t-i
    \end{pmatrix}\alpha^{K+2i-t}\beta^{t+\nu-2i}\\
    \nabla^\nu_{\Psi^-} \begin{pmatrix}
        K\\
        t
    \end{pmatrix}\alpha^{K-t}\beta^t &= \sum_{i=0}^\nu \varepsilon_i \begin{pmatrix}
        \nu\\
        i
    \end{pmatrix}
     \begin{pmatrix}
        K\\
        t-i
    \end{pmatrix}\alpha^{K+2i-t}\beta^{t+\nu-2i}
\end{align}
where $\varepsilon_i=(-1)^i$. Each operator commutes with $\nabla_{\Phi^+}$, but generally not with each other. From a probability-and-fidelity perspective, the order does not matter if the only outcome types are $\ket{\Phi^+}$ and one other, but it does matter more generally.

The probability functions, $p_{\nu,K}^{\Pi}$, of $K$ outcomes of $\ket{\Phi^+}$ and $\nu$ outcomes of $\ket{\Pi}$ (in any permutation) are
\begin{align}
    p_{\nu,K}^{\Pi} &= \sum_{t=0}^{K+\nu} \frac{(K+\nu-t)!t!}{2^{K+\nu} (K+\nu+1)!}\abs{\nabla^\nu_{\Pi} \begin{pmatrix}
        K\\
        t
    \end{pmatrix}\alpha^{K-t}\beta^t}^2\,.
\end{align}
We perform some non-exhaustive numerical analyses which suggest that, as $\nu$ increases, the average fidelity (unsurprisingly) decreases, yet the state-dependent fidelity remains maximal for the optimal Pauli eigenstates discussed in Sec.~\ref{sec: performances of unsuccessful outcomes} and Appendix~\ref{app: psi outcomes}. Furthermore, multiple unsuccessful measurements may result in an increase in confidence that $\ket{\varphi}_X$ is close to a specific eigenbasis, and so this particular line of inquiry may be insightful for the information gain analysis conducted in Sec.~\ref{sec: quantifying info gain}. One may also consider the effect that multiple types of unsuccessful measurements (e.g., $\ket{\Phi^-}$ and $\ket{\Psi^+}$) may have on the performance of the protocol. Continuing this investigation further is beyond the scope of this manuscript; we therefore leave this open for future works.

\section{Fidelity for outcomes $\ket{\Psi^\pm}_{X_kA_{k-1}}$}\label{app: psi outcomes}
\begin{figure*}
    \centering
    \includegraphics[width=\linewidth]{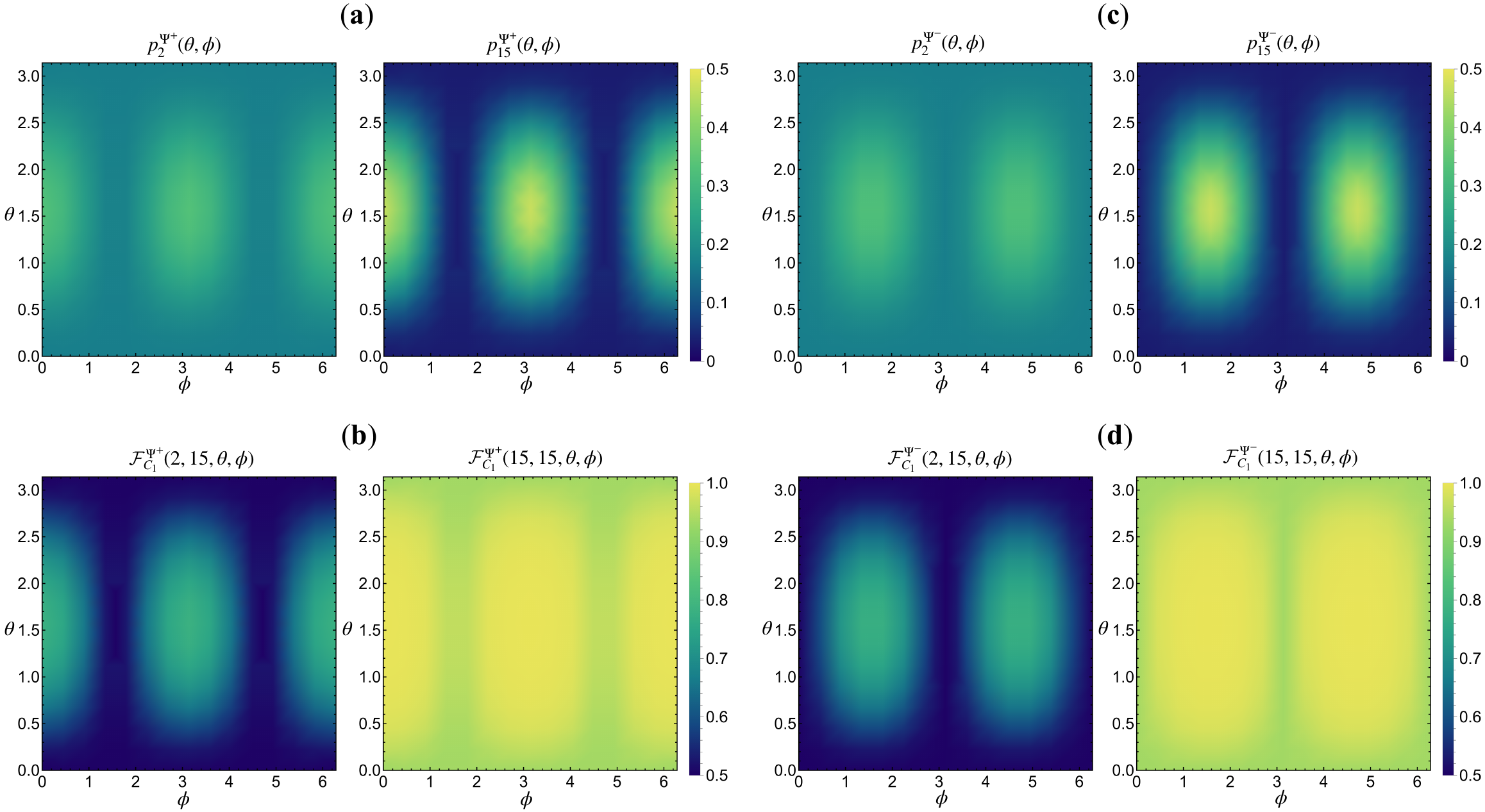}
    \caption{Heatmaps showing the probability and fidelity functions for the BSM outcomes $\ket{\Psi^\pm}_{X_kA_{k-1}}$ in a similar manner to Fig.~\ref{fig: phi minus heatmaps}. We show this for $k=2,15$ and $M=15$. \textbf{(a)} and \textbf{(b)} are respectively the probability and fidelity functions for outcome $\ket{\Psi^+}_{X_kA_{k-1}}$, while \textbf{(c)} and \textbf{(d)} are respectively the probability and fidelity functions for outcome $\ket{\Psi^-}_{X_kA_{k-1}}$.}
    \label{fig: psi heatmaps}
\end{figure*}
The state-dependent fidelity functions for the outcomes $\ket{\Psi^\pm}$ are
\begin{align}
    \mathcal{F}_{C_1}^{\Psi^+}(k,M,\theta,\phi) &= \frac{\big[k\mathcal{D}_{k,M}+2\big]\xi_x + (k+2)\mathcal{D}_{k,M}}
       {M(k+2)\big[(k-1)\xi_x + k+1\big]}\\
       \mathcal{F}_{C_1}^{\Psi^-}(k,M,\theta,\phi)&= \frac{\big[k\mathcal{D}_{k,M}+2\big]\xi_y + (k+2)\mathcal{D}_{k,M}}
       {M(k+2)\big[(k-1)\xi_y + k+1\big]}\,,
\end{align}
with $\xi_x=(2\sin^2\!\theta\cos^2\!\phi - 1)$ and $\xi_y=(2\sin^{2}\theta\sin^{2}\phi - 1)$. For a given $k$ and $M$, $\mathcal{F}_{C_1}^{\Psi^+}$ reaches a maximum value of $\gamma(k,M)$ at $\theta=\frac{\pi}{2}$ and $\phi=0,\pi$ (corresponding to the states $\ket{\pm}$), while minima equal to Eq.~\eqref{eq: minimum fidelity} are reached for the other Pauli eigenstates. A similar behavior can be observed for $\mathcal{F}_{C_1}^{\Psi^-}$, but with the maxima at $\theta=\frac{\pi}{2}$ and $\phi=\frac{\pi}{2},\frac{3\pi}{2}$ (corresponding to $\ket{\pm \ui}$) instead. We display the behaviors of the above functions (alongside their corresponding probability functions) with $k$ in Fig.~\ref{fig: psi heatmaps}, highlighting the similarities with Fig.~\ref{fig: phi minus heatmaps} but for different Pauli eigenbases.
\newpage
\end{widetext}

\end{document}